\documentclass[11pt]{article}

\usepackage{listings}
\usepackage{amsmath}
\usepackage{comment}
\usepackage[skins, breakable]{tcolorbox}
\usepackage{multirow}
\usepackage{amsthm}
\usepackage{amsfonts}
\usepackage{graphicx}
\usepackage{appendix}
\usepackage{xspace}
\usepackage{color}
\usepackage{tikz}
\usepackage{pgfplots}
\usepackage{textcomp}
\usepackage{enumitem}
\usepackage{url}
\usepackage{subcaption}
\usepackage[T1]{fontenc}
\usepackage{fullpage}
\usepackage{changepage}
\usepackage{authblk}
\usepackage{wrapfig}

\newtheorem{lemma}{Lemma}
\newtheorem{theorem}{Theorem}
\newtheorem{definition}{Definition}

\newcommand{\fc} {\mathcal{C}}

\newcommand{\ff} {\mathcal{F}}

\newcommand{\propose}{\texttt{propose}\xspace}
\newcommand{\vote}{\texttt{vote}\xspace}

\newcommand{\state}{\texttt{state}\xspace}
\newcommand{\timeout}{\texttt{timeout}\xspace}

\newcommand{\marker}{\textsl{marker}\xspace}

\newcommand{\strongvote}{\texttt{strong-vote}\xspace}
\newcommand{\strongqc}{\texttt{strong-QC}\xspace}
\newcommand{\strongvotes}{\texttt{strong-vote}s\xspace}
\newcommand{\strongqcs}{\texttt{strong-QC}s\xspace}

\newtcolorbox{mybox}[1][]{
enhanced,
colback=white,
boxsep=0pt,
#1
}

\title{Strengthened Fault Tolerance in Byzantine Fault Tolerant Replication}

\author[1]{Zhuolun Xiang\thanks{xiangzl@illinois.edu, work done at Novi}}
\author[2]{Dahlia Malkhi\thanks{{dmalkhi@diem.com}}}
\author[3]{Kartik Nayak\thanks{kartik@cs.duke.edu}}
\author[4]{Ling Ren\thanks{renling@illinois.edu}}
\affil[1,2]{Novi}
\affil[3]{Duke University}
\affil[1,4]{University of Illinois at Urbana-Champaign}

\date{}

\begin{document}

\maketitle

\begin{abstract}
Byzantine fault tolerant (BFT) state machine replication (SMR) is an important building block for constructing permissioned blockchain systems.
In contrast to Nakamoto Consensus where any block obtains higher assurance as buried deeper in the blockchain, in BFT SMR, any committed block is secure has a fixed resilience threshold.
In this paper, we investigate {\em strengthened fault tolerance (SFT)} in BFT SMR under partial synchrony, which provides gradually increased resilience guarantees (like Nakamoto Consensus) during an optimistic period when the network is synchronous and the number of Byzantine faults is small.
Moreover, the committed blocks can tolerate more than one-third (up to two-thirds) corruptions even after the optimistic period.
Compared to the prior best solution Flexible BFT which requires quadratic message complexity, our solution maintains the linear message complexity of state-of-the-art BFT SMR protocols and requires only marginal bookkeeping overhead.
We implement our solution over the open-source Diem project, and give experimental results that demonstrate its efficiency under real-world scenarios.
\end{abstract}

\section{Introduction}
\label{sec:introduction}

An extremely simple consensus protocol named Nakamoto consensus~\cite{nakamoto2008bitcoin} powers several of the world's largest crypto-currencies. In Nakamoto consensus, honest participants adopt and attempt to extend the longest proof-of-work (PoW) chain, and a block buried sufficiently deep in the blockchain is considered confirmed (with overwhelming probability).
An additional property of Nakamoto consensus is that, assuming honest majority computation power,
a block $B$ in the chain obtains higher assurance as the chain extending $B$ grows longer.
This provides a trade-off between transaction confirmation latency and safety, i.e., clients can choose to wait longer in exchange for more secure confirmations.

Differently, in permissioned consensus protocols, where the set of participants is fixed and known, transaction confirmation is binary: confirmation is safe and final so long as the number of malicious parties do not exceed a threshold of the system (one-third for asynchronous or partially synchronous networks~\cite{dwork1988consensus});
it is irreparably unsafe otherwise. 

It is thus natural to ask whether the safety assurance of a decision can be made to improve over time in permissioned protocols.
A recent algorithm called FBFT~\cite{malkhi2019flexible} provided such a solution (though it formulated the problem differently using a diverse client model).
In particular, a block can obtain resilience against a higher number of Byzantine failures, from one-third up to two-thirds, as it gathers additional confirmations (votes for the block).
We refer to this notion as {\em strengthened fault tolerance (SFT)}.
The term strengthened reflects the fact that stronger resilience guarantees are obtained given that the conditions are optimistic, e.g., when the network is synchronous and the number of Byzantine faults is small during the optimistic period. Once blocks are committed with higher resilience, they are safe even if the number of malicious parties later exceeds the one-third threshold.
Similar to the $k$-deep rule of Nakamoto consensus, with SFT the clients can choose to wait longer for valuable blocks to obtain higher resilience, trading off safety with latency.

While FBFT indirectly provided a solution to SFT, it did so on top of traditional BFT protocols (PBFT~\cite{castro1999practical}), and is hence not as simple and not ``chained'' or pipelined like Nakamoto consensus or recent chain-based BFT SMR protocols such as HotStuff~\cite{yin2019hotstuff}. 
More importantly, building the FBFT adaptability requires complex bookkeeping that maintain incremental confirmation, and incurs a quadratic message complexity overhead when applied to recent linear BFT protocols such as HotStuff~\cite{yin2019hotstuff,baudet2019state} (see Section~\ref{sec:diembft:oft}).

In this paper, we introduce a new approach for implementing SFT with several contributions. 
\textbf{First}, our approach works over the recent line of BFT consensus algorithms~\cite{buchman2016tendermint,buterin2017casper,yin2019hotstuff,baudet2019state,chan2020streamlet} that bring the chaining spirit and simplicity of Nakamoto consensus to the permissioned arena. 
\textbf{Second}, our approach induces marginal bookkeeping overhead and retains the {\em linear message complexity} of state-of-art BFT protocols.
\textbf{Third}, we implement our approach over DiemBFT (also known as LibraBFT earlier)~\cite{baudet2019state}  -- a production version of HotStuff that provides an infrastructure for efficient and inter-operable financial services. Our evaluation demonstrates that SFT-DiemBFT adds moderate latency increase to obtaining higher assurance guarantees (up to two-thirds) under various conditions. 
We also demonstrate our approach on another recent chain-based BFT called Streamlet~\cite{chan2020streamlet} in Appendix~\ref{sec:streamlet}, which may be of independent interests. 

The rest of the paper is organized as follows.
Section~\ref{sec:preliminary} introduces the definitions, an abstract prototype of chain-based BFT, and an overview of the DiemBFT protocol.
Section~\ref{sec:diembft} presents the formal definition of strengthened fault tolerance, and our solution for implementing SFT in DiemBFT.
Experimental evaluations of the protocol are presented in Section~\ref{sec:evaluation}. More discussions can be found in Section~\ref{sec:discussion}.
Finally, we summarize related works in Section~\ref{sec:relatedwork} and conclude in Section~\ref{sec:conclusion}.

\section{Preliminaries}
\label{sec:preliminary}

We consider a permissioned system of $n$ replicas numbered $1, 2, \ldots, n$. 
A public-key infrastructure exists to certify each party's public key.
Some replicas are Byzantine with arbitrary behaviors and controlled by an adversary; the rest of the replicas are called honest.
The replicas have all-to-all reliable and authenticated communication channels.
The system is assumed to be partially synchronous, i.e., there is a known network delay bound $\Delta$ that will hold after an unknown Global Stabilization Time (GST). After GST, all messages between honest replicas will arrive within time $\Delta$.
Without loss of generality, we let $n=3f+1$ where $f$ denotes the assumed upper bound on the number of Byzantine faults, which is the optimal worst-case resilience bound under partial synchrony~\cite{dwork1988consensus}.
Throughout the paper we will use $t$ to denote the {\em actual number} of Byzantine faults in the system, which may be less than $f$ during the optimistic period.
This paper aims to achieve fault tolerance guarantees beyond one-third optimistically -- blocks can be committed with higher assurance if optimistic conditions are met, and they will be secure against more than $f$ Byzantine faults after the optimistic periods.

\paragraph{Cryptographic primitives.}
We assume standard digital signatures and public-key infrastructure (PKI). We use $\langle x\rangle_i$ to denote a signed message $x$ by replica $i$.
We also assume a cryptographic hash function $H(\cdot)$ that can map an input of arbitrary size to an output of fixed size. The hash function is collision-resistant, which means the probability of an adversary producing $m\neq m'$ such that $H(m)=H(m')$ is negligible. 

\paragraph{Byzantine Fault Tolerant State Machine Replication~\cite{Abraham2020SyncHS}.}
A Byzantine fault tolerant state machine replication protocol commits client transactions as a linearizable log akin to a single non-faulty server, and provides the following two guarantees.
\begin{itemize}[leftmargin=*,noitemsep,topsep=0pt]
    \item Safety. Honest replicas do not commit different transactions at the same log position. 
    \item Liveness. Each client transaction is eventually committed by all honest replicas.
\end{itemize}

\smallskip
Besides the two requirements above, a validated BFT SMR protocol also requires to satisfy the {\em external validity} which requires any committed transactions to be {\em externally valid}, satisfying some application-dependent predicate. This can be done by adding validity checks on the transactions before the replicas proposing or voting, and for brevity we omit the details and focus on the BFT SMR formulation defined above.
For most of the paper, we omit the client from the discussion and only focus on replicas.
The discussion on how to prove SFT to the clients is provided in Section~\ref{sec:discussion}.

\subsection{Chain-based BFT SMR}
\label{sec:preliminary:prototype}

In a chain-based BFT, transactions are batched into {\em blocks} with some predefined ordering, and the protocol builds a chain of blocks to achieve a total ordering of transactions. 
In this section, we present relevant definitions and a prototype that abstracts all existing chain-based BFT SMR protocols~\cite{buchman2016tendermint,buterin2017casper,yin2019hotstuff,baudet2019state,chan2020streamlet}.
First, we define several terminologies commonly used in the literature.

\begin{itemize}[leftmargin=*,noitemsep,topsep=0pt]
    \item {\bf Quorum Certificate.}
    A quorum certificate (QC) is a set of signed votes for a block by a quorum of replicas, from $n-f=2f+1$ (out of $n=3f+1$) distinct replicas.
    We say a block is certified if there exists a QC for the block. QCs or certified blocks are ranked higher with larger round numbers.

    \item {\bf Block Format.}
    The position of a block in the chain is called height.
    A block at height $k$ is formatted as follows $B_k=(H(B_{k-1}), qc, txn)$ where $H(B_{k-1})$ is the hash digest of its parent block $B_{k-1}$ (at height $k-1$), $qc$ is a quorum certificate of the parent block $B_{k-1}$, and $txn$ is a batch of new transactions.
    
    \item {\bf Block Chaining.}
    Blocks are chained by hash digest and quorum certificates. 
    We say a block $B_l$ extends another block $B_k$ ($l \geq k$), if $B_k$ is an ancestor of $B_l$ in the chain.
    We say two blocks $B_k$ and $B_l$ are conflicting if they do not extend one another. 
\end{itemize}

With the terminologies above, the safety and liveness requirement of the BFT SMR protocol can be redefined as follows.
For safety, no two honest replicas commit different blocks at the same height.
For liveness, all honest replicas keep committing new blocks containing the clients' transactions after GST.

\paragraph{A prototype for chain-based BFT SMR.}

\begin{figure*}[h!]
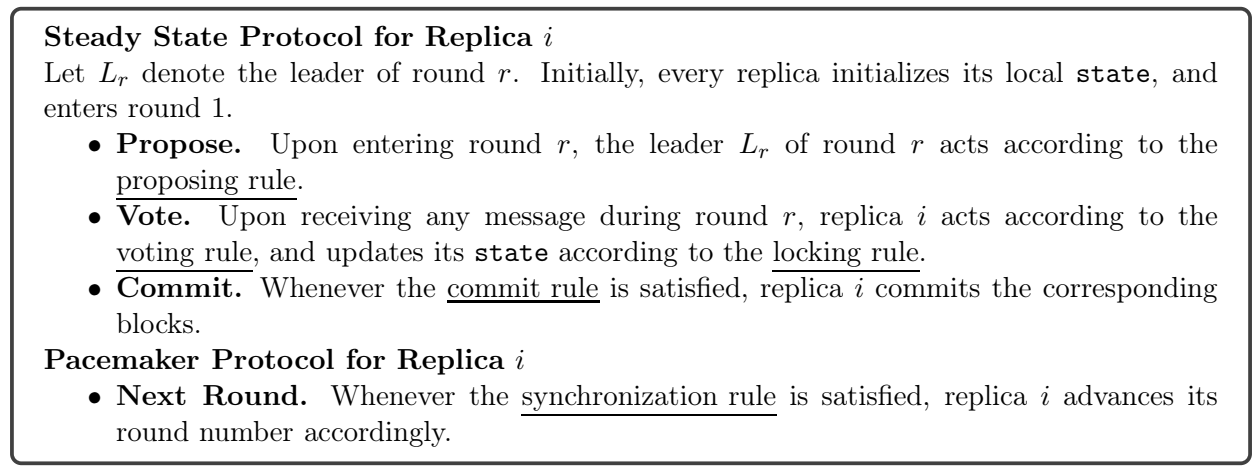

    \centering
    \begin{mybox}
\textbf{Steady State Protocol for Replica $i$}

Let $L_r$ denote the leader of round $r$. Initially, every replica initializes its local \state, and enters round $1$.
\begin{itemize}[noitemsep,topsep=0pt]
    \item\label{smr:propose} {\bf Propose.} 
    Upon entering round $r$,  the leader $L_r$ of round $r$ acts according to the \underline{proposing rule}.
    
    \item\label{smr:vote} {\bf Vote.} 
    Upon receiving any message during round $r$, replica $i$ acts according to the \underline{voting rule}, and updates its \state according to the \underline{locking rule}.
    
    \item\label{smr:commit} {\bf Commit.} 
    Whenever the \underline{commit rule} is satisfied, replica $i$ commits the corresponding blocks.
    
\end{itemize}
    
\textbf{Pacemaker Protocol for Replica $i$}
\begin{itemize}[noitemsep,topsep=0pt]
    \item\label{smr:nextround} {\bf Next Round.}
    Whenever the \underline{synchronization rule} is satisfied, replica $i$ advances its round number accordingly.
    
    \end{itemize}

    \end{mybox}

    \caption{Prototype for Chain-based BFT SMR.}
    \label{fig:prototype}
\end{figure*}

\begin{figure*}[h!]
    \centering
    \begin{mybox}
The DiemBFT protocol can be defined by the prototype in Figure~\ref{fig:prototype} with the following components. 

\begin{itemize}[noitemsep,topsep=0pt]
    \item \underline{\state.}
    Each replica locally keeps a highest voted round $r_{vote}$, a highest locked round $r_{lock}$, a current round $r_{cur}$ (all initialized to $0$), and a highest quorum certificate $qc_{high}$ (initialized to $\bot$ of round $0$).
    
    \item \underline{Proposing rule.}
    The leader multicasts a proposal $\langle \propose, B_k, r\rangle_{L_r}$ containing block $B_k=(H(B_{k-1}), qc_{high}, txn)$ extending the highest certified block $B_{k-1}$ (certified by $qc_{high}$) known to the leader.
    
    \item \underline{Voting rule.} 
    Upon receiving the first valid proposal $\langle \propose, B_k, r\rangle_{L_r}$ at round $r$,
    send a \vote message in the form of $\langle \vote, B_k, r \rangle_i$ to the next leader $L_{r+1}$, iff
    (1) $r>r_{vote}$, and
    (2) $B_{k-1}.round\geq r_{lock}$,
    where $B_{k-1}.round$ is the round number of $B_k$'s parent block $B_{k-1}$.
    
    \item \underline{Locking rule.}
    Upon receiving a valid block $B$, the replica updates $r_{vote}=\max(r_{vote},B.round)$ if it votes for $B$.
    Upon receiving a valid $qc$, let $B$ be the parent block of the block certified by $qc$, the replica updates $r_{lock} = \max( r_{lock}, B.round )$, and it also updates $qc_{high}=qc$ if $qc$ ranks higher than $qc_{high}$.
    
    \item \underline{Commit rule ($3$-chain rule).}
    The replica commits a block $B_k$ and all its ancestors, iff there exists three adjacent certified blocks $B_k,B_{k+1},B_{k+2}$ in the chain with consecutive round numbers, i.e.,
    $B_{k+2}.round=B_{k+1}.round+1=B_k.round+2$.
    
    \item \underline{Synchronization rule.}
    After receiving the QC ($2f+1$ distinct votes) of a round-$(r-1)$ block or $2f+1$ \timeout messages (explained below) of round $r-1$, update the current round $r_{cur}=\max(r_{cur}, r)$.
    
    \item \underline{Timeout.}
    Upon entering round $r$, the replica stops any timer or voting for round $<r$, and sets a timer $T_r$ for round $r$ to some predefined duration.
    Upon the timer $T_r$ expires, the replica stops voting for round $r$ and multicasts a \timeout message $\langle \timeout, r, qc_{high} \rangle_i$ of round $r$.
\end{itemize}

    \end{mybox}

    \caption{DiemBFT Protocol.}
    \label{fig:diembft}
\end{figure*}

Now we present a prototype that abstracts the existing protocols~\cite{buchman2016tendermint,buterin2017casper,yin2019hotstuff,baudet2019state,chan2020streamlet}, as shown in Figure~\ref{fig:prototype}.
The goal here is to help readers get a high-level picture of the protocol rather than formalizing a universal framework for chain-based BFT SMR.
The protocol executes in rounds $r=1,2,\dots$, each with a designated replica as the round leader. This paper assumes a round-robin rotation for leader elections, which is also assumed in \cite{yin2019hotstuff,baudet2019state,chan2020streamlet}. 
The protocol consists of two components, a Steady State protocol that aims to make progress when the round leader is honest, and a Pacemaker protocol that advances the round number either due to the lack of progress or the current round being completed. 

The Steady State protocol at each replica involves several rounds of message exchanges. First, the current round leader executes according to the protocol-specific \underline{proposing rule}, such as proposing a block.
When replicas receive the leader's proposal, they will take actions according to the protocol-specific \underline{voting rule}, most often sending vote messages to the leader or other replicas. The replicas also update their local \state according to the \underline{locking rule}.
When any replica observes that certain conditions described by the \underline{commit rule} are satisfied, they commit the corresponding blocks (and all ancestor blocks due to chaining).
The Pacemaker protocol at each replica is responsible for advancing the round numbers according to the protocol-specific \underline{synchronization rule}. 

State-of-art chain-based BFT protocols like DiemBFT and Streamlet can be instantiated by specifying the \state and the missing protocol rules, as will be described in Section~\ref{sec:diembft:overview} and Appendix~\ref{sec:streamlet:overview}.

\begin{figure*}[h!]
    \centering

\resizebox{\textwidth}{!}{

\tikzset{every picture/.style={line width=0.75pt}} 

\begin{tikzpicture}[x=0.75pt,y=0.75pt,yscale=-1,xscale=1]

\draw   (331,80.67) -- (371.2,80.67) -- (371.2,119.6) -- (331,119.6) -- cycle ;
\draw   (411,80.67) -- (451.2,80.67) -- (451.2,119.6) -- (411,119.6) -- cycle ;
\draw   (491,80.67) -- (531.2,80.67) -- (531.2,119.6) -- (491,119.6) -- cycle ;
\draw   (571,80.67) -- (611.2,80.67) -- (611.2,119.6) -- (571,119.6) -- cycle ;
\draw  [fill={rgb, 255:red, 0; green, 0; blue, 0 }  ,fill opacity=1 ] (391,93) -- (396,100) -- (391,107) -- (386,100) -- cycle ;
\draw    (371.5,100) -- (411.5,100) ;
\draw  [fill={rgb, 255:red, 0; green, 0; blue, 0 }  ,fill opacity=1 ] (471,93) -- (476,100) -- (471,107) -- (466,100) -- cycle ;
\draw    (451.5,100) -- (491.5,100) ;
\draw  [fill={rgb, 255:red, 0; green, 0; blue, 0 }  ,fill opacity=1 ] (551,93) -- (556,100) -- (551,107) -- (546,100) -- cycle ;
\draw    (531.5,100) -- (570.5,100) ;
\draw    (550.5,86) -- (550.5,58) ;
\draw [shift={(550.5,56)}, rotate = 450] [color={rgb, 255:red, 0; green, 0; blue, 0 }  ][line width=0.75]    (10.93,-3.29) .. controls (6.95,-1.4) and (3.31,-0.3) .. (0,0) .. controls (3.31,0.3) and (6.95,1.4) .. (10.93,3.29)   ;
\draw   (19,81.67) -- (59.2,81.67) -- (59.2,120.6) -- (19,120.6) -- cycle ;
\draw   (99,81.67) -- (139.2,81.67) -- (139.2,120.6) -- (99,120.6) -- cycle ;
\draw   (179,81.67) -- (219.2,81.67) -- (219.2,120.6) -- (179,120.6) -- cycle ;
\draw  [fill={rgb, 255:red, 0; green, 0; blue, 0 }  ,fill opacity=1 ] (79,94) -- (84,101) -- (79,108) -- (74,101) -- cycle ;
\draw    (59.5,101) -- (99.5,101) ;
\draw  [fill={rgb, 255:red, 0; green, 0; blue, 0 }  ,fill opacity=1 ] (159,94) -- (164,101) -- (159,108) -- (154,101) -- cycle ;
\draw    (139.5,101) -- (179.5,101) ;
\draw    (158.5,82) -- (158.5,54) ;
\draw [shift={(158.5,52)}, rotate = 450] [color={rgb, 255:red, 0; green, 0; blue, 0 }  ][line width=0.75]    (10.93,-3.29) .. controls (6.95,-1.4) and (3.31,-0.3) .. (0,0) .. controls (3.31,0.3) and (6.95,1.4) .. (10.93,3.29)   ;

\draw (351.1,100.13) node  [font=\large]  {$B_{k}$};
\draw (431.1,100.13) node  [font=\large]  {$B_{k+1}$};
\draw (511.1,100.13) node  [font=\large]  {$B_{k+2}$};
\draw (591.1,100.13) node  [font=\large]  {$B_{k+3}$};
\draw (378,122) node [anchor=north west][inner sep=0.75pt]  [font=\large] [align=left] {$\displaystyle QC_{k}$};
\draw (454,122) node [anchor=north west][inner sep=0.75pt]  [font=\large] [align=left] {$\displaystyle QC_{k+1}$};
\draw (535,122) node [anchor=north west][inner sep=0.75pt]  [font=\large] [align=left] {$\displaystyle QC_{k+2}$};
\draw (440,25) node [anchor=north west][inner sep=0.75pt]  [font=\large] [align=left] {Commit $\displaystyle B_{k}$ when receiving $\displaystyle QC_{k+2}$.};
\draw (323,151) node [anchor=north west][inner sep=0.75pt]  [font=\large] [align=left] {round $\displaystyle r$};
\draw (390,151) node [anchor=north west][inner sep=0.75pt]  [font=\large] [align=left] {round $\displaystyle r+1$};
\draw (480,151) node [anchor=north west][inner sep=0.75pt]  [font=\large] [align=left] {round $\displaystyle r+2$};
\draw (39.1,101.13) node  [font=\large]  {$B_{k-2}$};
\draw (119.1,101.13) node  [font=\large]  {$B_{k-1}$};
\draw (199.1,101.13) node  [font=\large]  {$B_{k}$};
\draw (63,123) node [anchor=north west][inner sep=0.75pt]  [font=\large] [align=left] {$\displaystyle QC_{k-2}$};
\draw (143,123) node [anchor=north west][inner sep=0.75pt]  [font=\large] [align=left] {$\displaystyle QC_{k-1}$};
\draw (0,25) node [anchor=north west][inner sep=0.75pt]  [font=\large] [align=left] {Update $\displaystyle r_{lock} =max( r_{lock} ,r')$ when receiving $\displaystyle QC_{k-1}$.};
\draw (11,152) node [anchor=north west][inner sep=0.75pt]  [font=\large] [align=left] {round $\displaystyle r'$};
\draw (173,152) node [anchor=north west][inner sep=0.75pt]  [font=\large] [align=left] {round $\displaystyle r$};
\draw (50,181) node [anchor=north west][inner sep=0.75pt]  [font=\large] [align=left] {$2$-chain locking rule};
\draw (400,180) node [anchor=north west][inner sep=0.75pt]  [font=\large] [align=left] {$3$-chain commit rule};

\end{tikzpicture}
    
}
    
    \caption{Illustration of DiemBFT.}
    \label{fig:diemfbt:illustration}
\end{figure*}
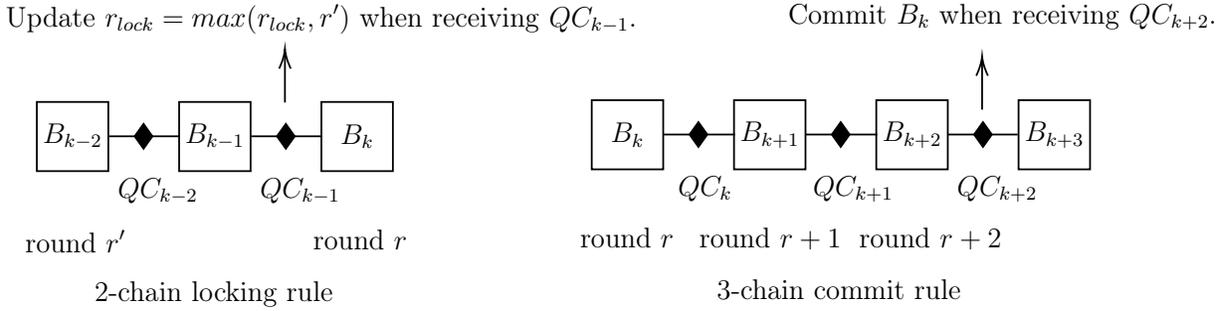

\subsection{Overview of DiemBFT}
\label{sec:diembft:overview}
The DiemBFT protocol (known as LibraBFT earlier)~\cite{baudet2019state} is a production version of HotStuff~\cite{yin2019hotstuff} with an open-source implementation~\cite{diem}.
DiemBFT provides responsiveness (advance with network speed), amortized linear communication complexity per block decision and amortized constant round commit latency under synchrony.
Given the abstract prototype in Figure~\ref{fig:prototype}, we can present the DiemBFT protocol by specifying the missing components, as shown in Figure~\ref{fig:diembft}.
The leader always proposes a block extending the highest certified block known to the leader. When receiving the leader's proposal, replicas send a \vote message containing the block (in practice only the hash digest of the block) to the next leader, if the proposal's round number is larger than its highest voted round number $r_{vote}$ and the proposal extends a block of a round number $\geq r_{lock}$.
The local variables $r_{vote}$, $r_{lock}$ and the highest QC will be updated whenever the replica receives a new block or a new QC (from the proposal blocks or \timeout messages). 
When the replica observes a $3$-chain that contains three certified blocks with consecutive round numbers, it commits the first block in the $3$-chain and all ancestor blocks (this is known as the $3$-chain rule from HotStuff~\cite{yin2019hotstuff}).
Replicas send \timeout messages containing the highest QC $qc_{high}$ if they detect the current round leader is not making progress.
The replica $i$ 
will advance the current round number either when a new QC is received (contained in proposal or formed after collecting \vote messages) or when $2f+1$ \timeout messages are collected. 
More details on pacemaker and synchronization in BFT protocols can be found in \cite{baudet2019state,naor2019cogsworth,naor2020expected}.

\section{DiemBFT with Strengthened Fault Tolerance}
\label{sec:diembft}
We propose a chain-based BFT SMR protocol that provides strengthened fault tolerance based on the DiemBFT protocol~\cite{baudet2019state}.
We first formally define strengthened fault tolerance in Section~\ref{sec:preliminary:definition}, and then present our solution for implementing strengthened fault tolerance called SFT-DiemBFT in Section~\ref{sec:diembft:oft}.
Our SFT-DiemBFT is the first solution that achieves SFT while keeping the protocol responsive and linear. In constrast, adapting FBFT~\cite{malkhi2019flexible} to DiemBFT to obtain SFT incurs a quadratic message complexity overhead per block decision, as will be explained in Section~\ref{sec:diembft:oft}.

\subsection{Strengthened Fault Tolerance}
\label{sec:preliminary:definition}
Strengthened fault tolerance aims at providing higher assurance on blocks, i.e., only one block is committed with the said higher assurance at every height of the chain even if the adversary corrupts more than $f$ replicas in the future.
More specifically, a strengthened fault tolerant protocol $\Pi_{\text{sft}}$ defines a set of {\em strong commit} rules, that can commit blocks with different levels of fault tolerance guarantees.
A block is {\em $x$-strong committed} if the commit is secure against up to $x$ Byzantine faults, i.e., no other conflicting blocks can be $x$-strong committed if the number of Byzantine faults does not exceed $x$.
In contrast, traditional BFT protocols have a single commit rule that is $f$-strong (i.e., tolerating $f$ faults), which we refer to as the {\em regular commit} rule.
We will present SFT solutions for DiemBFT and Streamlet where blocks can be $x$-strong committed with $x\in [f,2f]$. 

The safety of an SFT protocol is formally defined as follows.

\begin{definition}[Safety]
    A strengthened fault tolerant protocol $\Pi_{\text{sft}}$ with a set of $x$-strong commit rules is safe, if it satisfies the following: 
    For any $x$-strong committed block $B$, no other conflicting block $B'$ can be $x'$-strong committed where $x'\geq x$ in the presence of $t\leq x$ Byzantine faults. 
\end{definition}

For instance, if a block $B$ is $(f+1)$-strong committed, even if the adversary corrupts $f+1$ replicas, no conflicting block can be $x$-strong committed for any $x\geq f+1$. However, the $(f+1)$-strong commit of $B$ does not rule out a conflicting $f$-strong committed block $B'$, as the assumption of the $f$-strong commit is violated (number of faults is $f+1$) and the $f$-strong guarantee for $B'$ does not hold anymore.

These strong commits (for strength $x>f$) are optimistic in the sense that they only occur during optimistic periods with fewer faults. 
Hence, we also modify the {\em liveness} requirement to ensure that a block will be strong committed during the optimistic period.
Throughout the paper, we will use $t$ to denote the actual number of Byzantine faults in the system during the optimistic period.

\begin{definition}[Liveness]
    An strengthened fault tolerant protocol $\Pi_{\text{sft}}$ is live, if during an optimistic period of time with $t\leq f$ Byzantine faults, a block will be $(2f-t)$-strong committed.
\end{definition}

Note that in the liveness definition above, the definition {\em optimistic period} can be protocol-specific.
In Section~\ref{sec:diembft:oft} we first present a simple solution where the optimistic period requires benign (crash) faults (Theorem~\ref{thm:liveness}).
Note that if the optimistic period has no failures, blocks will be $2f$-strong committed, thereby tolerating up to $2f$ Byzantine faults in the future.
In Section~\ref{sec:discussion:strongvote}, we extend the solution to a stronger liveness guarantee where the optimistic period allows Byzantine faults (Theorem~\ref{thm:newliveness}) with some added overhead. 

\subsection{SFT-DiemBFT}
\label{sec:diembft:oft}

A natural attempt to obtain SFT is to simply exchange more votes among replicas, i.e., increasing the QC size of the $3$-chain rule.
This is the approach taken by the FBFT protocol~\cite{malkhi2019flexible}.
FBFT is built on the classic PBFT protocol, which is not chain-based and has quadratic message complexity.
While it can be adapted to be chain-based (more details can be found in Appendix~\ref{sec:fbft}), we do not see a way to avoid its quadratic message complexity.
The reason is that a leader of a round cannot wait indefinitely for a QC of size larger than $2f+1$, since that will break liveness in the presence of $f$ Byzantine faults. 
Thus, replicas need to advance to the next round with a standard quorum size of $2f+1$, and any extra votes from earlier rounds need to be multicasted later by the leader of that round.
Yet, the leader did not know how many extra votes can be collected, and hence cannot wait indefinitely for any predefined number of extra votes either.
This means the leader needs to multicast up to $f$ extra votes one by one as they arrive in the worst case, leading to an amortized message complexity of $O(fn)=O(n^2)$ per round/decision.

Our solution is inspired by the Nakamoto consensus where successor blocks in the chain extension increase the assurance of earlier blocks.
We say a vote message for block $B$ is a \emph{direct} vote for $B$, and is an \emph{indirect} vote for any ancestor block of $B$.
We make a similar observation in the chain-based BFT SMR that {\em indirect votes in the QCs of a block's extension blocks can help increase resilience}. 

However, simply counting all the indirect votes towards the resilience may overcalculate the actual fault tolerance.
Intuitively, if a replica had previously voted for blocks in another fork, and then voted for a block $B$ in the current branch, the vote for $B$ may not enhance the resilience of all the ancestor blocks in this branch since the replica already contributed to a conflicting fork. We give a concrete illustrating example in Appendix~\ref{sec:diembft:quorumdiversity}.

Thus, to ensure correctness, SFT-DiemBFT identifies the set of indirect votes that actually increase the assurance of a block, by including some additional information in each vote message that summarizes a replica's voting history.
We first present a simpler solution that additionally attaches just one value called \marker in the vote message, but with a relatively weak liveness guarantee on the strong commits. Later in Section~\ref{sec:discussion:strongvote}, we show how SFT-DiemBFT can provide better liveness guarantees by including more information in the votes.
Such a vote message with additional information attached is called a \strongvote.

\begin{figure*}[t!]
    \centering
    \begin{mybox}
The SFT-DiemBFT protocol modifies the original DiemBFT protocol in Figure~\ref{fig:diembft} with the following changes.

\begin{itemize}
    \item {\bf Changes to local \state.}
    For every fork in the blockchain, the replica additionally keeps the highest voted block on that fork.
    
    \item {\bf Strong-vote and strong-QC.}
    For a \vote message of any block $B$, the replica additionally includes a $\marker=\max\{B'.round ~|~ \text{$B'$ conflicts $B$ and replica voted for $B'$}\}$ in the \vote.
    Therefore, the \vote message becomes $\langle \vote, B, r, \marker \rangle_i$, and we refer the \vote message defined above as \strongvote.
    Each block now contains a \strongqc consisting of $2f+1$ distinct signed \strongvote message of its parent block.

    \item {\bf Endorsements.}
    A \strongvote $\langle \vote, B', r', \marker \rangle_i$ endorses a round-$r$ block $B$, if $B=B'$, or $B'$ extends $B$ and $\marker<r$.
    We say a replica is an endorser of block $B$, if its \strongvote endorses $B$.

    \item {\bf Strong commit rule (strong $3$-chain rule).}
    The replica $x$-strong commits a block $B_k$ and all its ancestors, if and only if there exists three adjacent blocks $B_k,B_{k+1},B_{k+2}$ in the chain with consecutive round numbers, 
    and $B_k,B_{k+1},B_{k+2}$ all have \underline{at least $x+f+1$ endorsers}.
\end{itemize}

    \end{mybox}
    \caption{Strengthened Fault Tolerance for DiemBFT.}
    \label{fig:oft-diembft}
\end{figure*}

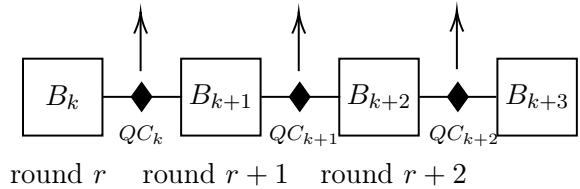
\begin{figure}[t!]
    \centering

\resizebox{.75\textwidth}{!}{

\tikzset{every picture/.style={line width=0.75pt}} 

\begin{tikzpicture}[x=0.75pt,y=0.75pt,yscale=-1,xscale=1]

\draw   (193,92.67) -- (233.2,92.67) -- (233.2,131.6) -- (193,131.6) -- cycle ;
\draw   (273,92.67) -- (313.2,92.67) -- (313.2,131.6) -- (273,131.6) -- cycle ;
\draw   (353,92.67) -- (393.2,92.67) -- (393.2,131.6) -- (353,131.6) -- cycle ;
\draw   (433,92.67) -- (473.2,92.67) -- (473.2,131.6) -- (433,131.6) -- cycle ;
\draw  [fill={rgb, 255:red, 0; green, 0; blue, 0 }  ,fill opacity=1 ] (253,105) -- (258,112) -- (253,119) -- (248,112) -- cycle ;
\draw    (233.5,112) -- (273.5,112) ;
\draw  [fill={rgb, 255:red, 0; green, 0; blue, 0 }  ,fill opacity=1 ] (333,105) -- (338,112) -- (333,119) -- (328,112) -- cycle ;
\draw    (313.5,112) -- (353.5,112) ;
\draw  [fill={rgb, 255:red, 0; green, 0; blue, 0 }  ,fill opacity=1 ] (413,105) -- (418,112) -- (413,119) -- (408,112) -- cycle ;
\draw    (393.5,112) -- (432.5,112) ;
\draw    (412.5,98) -- (412.5,70) ;
\draw [shift={(412.5,68)}, rotate = 450] [color={rgb, 255:red, 0; green, 0; blue, 0 }  ][line width=0.75]    (10.93,-3.29) .. controls (6.95,-1.4) and (3.31,-0.3) .. (0,0) .. controls (3.31,0.3) and (6.95,1.4) .. (10.93,3.29)   ;
\draw    (332.5,99) -- (332.5,71) ;
\draw [shift={(332.5,69)}, rotate = 450] [color={rgb, 255:red, 0; green, 0; blue, 0 }  ][line width=0.75]    (10.93,-3.29) .. controls (6.95,-1.4) and (3.31,-0.3) .. (0,0) .. controls (3.31,0.3) and (6.95,1.4) .. (10.93,3.29)   ;
\draw    (252.5,99) -- (252.5,71) ;
\draw [shift={(252.5,69)}, rotate = 450] [color={rgb, 255:red, 0; green, 0; blue, 0 }  ][line width=0.75]    (10.93,-3.29) .. controls (6.95,-1.4) and (3.31,-0.3) .. (0,0) .. controls (3.31,0.3) and (6.95,1.4) .. (10.93,3.29)   ;

\draw (213.1,112.13) node    {$B_{k}$};
\draw (293.1,112.13) node    {$B_{k+1}$};
\draw (373.1,112.13) node    {$B_{k+2}$};
\draw (453.1,112.13) node    {$B_{k+3}$};
\draw (240,125) node [anchor=north west][inner sep=0.75pt]  [font=\scriptsize] [align=left] {$\displaystyle QC_{k}$};
\draw (316,125) node [anchor=north west][inner sep=0.75pt]  [font=\scriptsize] [align=left] {$\displaystyle QC_{k+1}$};
\draw (397,125) node [anchor=north west][inner sep=0.75pt]  [font=\scriptsize] [align=left] {$\displaystyle QC_{k+2}$};
\draw (132,41) node [anchor=north west][inner sep=0.75pt]   [align=left] {$\displaystyle x$-strong commit $\displaystyle B_{k}$ when $\displaystyle B_{k} ,B_{k+1} ,B_{k+2}$ all have $\displaystyle x+f+1$ endorsers };
\draw (185,144) node [anchor=north west][inner sep=0.75pt]   [align=left] {round $\displaystyle r$};
\draw (252,144) node [anchor=north west][inner sep=0.75pt]   [align=left] {round $\displaystyle r+1$};
\draw (342,144) node [anchor=north west][inner sep=0.75pt]   [align=left] {round $\displaystyle r+2$};

\end{tikzpicture}
    
}
    
    \caption{Illustration for Strong Commit}
    \label{fig:diembft:strongcommit}
\end{figure}

\paragraph{Protocol Description.}
The protocol SFT-DiemBFT is presented in Figure~\ref{fig:oft-diembft}. 
As mentioned, the first step is to strengthen the \vote messages for block $B$ by attaching a \marker, which equals the largest round number of any block that voted by the replica and conflicts with $B$.
The default value of \marker is $0$ if the genesis block has round $1$.
A \strongvote with a \marker for block $B'$ can ``vote for'' block $B$, if $B=B'$ or the round number of $B$ is larger than the value of $\marker$ and $B'$ extends $B$.
We say the above \strongvote endorses the block $B$, and the sender of the \strongvote is an {\em endorser} of the block $B$.
To strong commit a block $B$ with higher confidence, the SFT-DiemBFT protocol uses a {\em strong $3$-chain rule} as the strong commit rule, which requires every block in the $3$-chain to have a set of $x+f+1$ endorsers, instead of $2f+1$ direct votes from a QC in the regular commit rule.
Note that the regular commit rule is a special case of the strong commit rule with $x=f$, i.e., three consecutive blocks each with $2f+1$ endorsers result in a tolerance of the regular commit rule is $2f+1-f-1=f$.

In SFT-DiemBFT, a block $B$ can be strong committed with  increased fault tolerance guarantees as the number of endorsers of every block in the $3$-chain increases. 
As the chain extending $B$ grows, there may be more \strongqc containing \strongvote that endorses $B$, and thus the number of endorsers of $B$ increases. 
Hence the level of resilience of $B$ is likely to increase as the chain extending $B$ grows longer.

In terms of efficiency, SFT-DiemBFT only induces marginal overhead over the original DiemBFT protocol -- only one \marker included in the votes.
Therefore, SFT-DiemBFT maintains all the desirable properties of DiemBFT, including responsiveness (protocol advances in network speed), amortized linear message complexity per block decision and amortized constant round commit latency during synchronous periods.
Our experiments in Section~\ref{sec:evaluation} confirm that the throughput and regular commit latency of SFT-DiemBFT are mostly identical to DiemBFT; in addition, blocks obtain higher resilience gradually during optimistic periods.

\subsection{Proof of Correctness}
\label{sec:diembft:proof}

We use $\fc(B)$ to denote the set of replicas that certifies $B$, and hence $|\fc(B)|=2f+1$.

\begin{lemma}\label{lem:1}
    If a block $B$ with round number $r$ has $E$ endorsers, then no other conflicting block with round number $r$ can be certified under $t\leq E-f-1$ Byzantine replicas.
\end{lemma}
\begin{proof}
    Suppose for the sake of contradiction, let block $B'$ be a conflicting certified block with  round number $r$. 

    The honest replicas in $\fc(B')$ will not vote for $B$, since honest replicas only vote for one block at each round.
    The replicas in $\fc(B')$ may vote for blocks extending $B$, as they are locked on some round number $<r$, and are able to vote for any blocks of round number $> r$ extending $B$ of round number $r$.
    However, by the definition of endorse, the honest replicas in $\fc(B')$ do not endorse block $B$, since they voted for $B'$ and have $\marker\geq r$. 
    Therefore, the intersection of $\fc(B')$ and the set of endorsers of $B$ contains no honest replicas, but only Byzantine replicas.
    Since the size of the intersection is at least $E+|\fc(B')|-n=E-f$, there should be at least $E-f$ Byzantine replicas, which contradicts the assumptions that there are at most $t\leq E-f-1$ Byzantine replicas.
    Therefore, any conflicting block $B'$ cannot be certified at round $r$.
\end{proof}

\begin{lemma}\label{lem:2}
    If a block $B$ of round number $r$ is $x$-strong committed at some honest replica by the strong commit rule and the number of Byzantine faults is $\leq x$, then any certified block with round number $\geq r$ must extend $B$.
\end{lemma}

\begin{proof}
    Since block $B$ is $x$-strong committed, there exists three blocks $B=B_k,B_{k+1},B_{k+2}$ with consecutive round numbers $r,r+1,r+2$, and each of the block $B_k,B_{k+1},B_{k+2}$ has $\geq x+f+1$ endorsers.
    Suppose the lemma is not true for the sake of contradiction, and let block $B'$ be a conflicting certified block with the smallest round number $r'\geq r$.
    Then, the parent block of $B'$ has round number $<r$, and does not extend $B_k$.
    Consider the following cases.
    \begin{itemize}
        \item When $r\leq r'\leq r+2$.  
        By Lemma \ref{lem:1}, for $i=0,1,2$, since $B_{k+i}$ has $x+f+1$ endorsers, no conflicting block with round number $r+i$ can be certified under $t\leq x$ Byzantine replicas.
        
        \item When $r'\geq r+3$. 
        The honest replicas in the set $\fc(B_{k+2})$ lock on round number $r$ after voting for $B_{k+2}$.
        Thus, they will not vote for block $B'$ since the parent block of $B'$ has round number $<r$.
        Let $\mathcal{B}$ denote the set of blocks of round number $\leq r'$ that extends $B_{k+2}$.
        Then the honest replicas in $\bigcup_{B_j\in\mathcal{B}} \fc(B_j)$ do not vote for $B'$ since they are locked on round number $\geq r$.
        Replicas in $\fc(B')$ may vote for blocks of round number $>r'$ that extend $B_{k+2}$.
        However, by the definition of endorse, the honest replicas in $\fc(B')$ do not endorse block $B_{k+2}$, since they voted for $B'$ and have $\marker\geq r+2$.
        Therefore, the intersection of $\fc(B')$ and the set of endorsers of $B_{k+2}$ contains no honest replicas, but only Byzantine replicas. 
        Since the size of intersection is at least $x+f+1+|\fc(B')|-n=x+1$, there should be at least $x+1$ Byzantine replicas, which contradicts the assumptions that there are at most $t\leq x$ Byzantine replicas.
        Therefore, any conflicting block $B'$ cannot be certified at round $r' \geq r+3$.
        
    \end{itemize}
\end{proof}

\begin{theorem}[Safety]\label{thm:safety}
    SFT-DiemBFT with the set of $x$-strong commit rules defined in Figure~\ref{fig:oft-diembft} is safe.
\end{theorem}

\begin{proof}
    Suppose $B$ is $x$-strong committed due to some block $B_k$ being $x$-strong committed directly.
    Suppose the safety is violated for the sake of contradiction, under $t\leq x$ Byzantine replicas, a conflicting block $B'$ is also $x'$-strong committed  due to block $B_{k'}'$ being $x'$-strong committed directly where $x'\geq x$.
    Suppose $B_k$  has round number $r$, and $B_{k'}'$ has round number $r'$.
    Without loss of generality, suppose that $r'\geq r$.
    By Lemma \ref{lem:2}, under $t\leq x$ faults any certified block with round number $\geq r$ must extend $B_k$. 
    Since $B,B'$ are conflicting blocks, $B_{k'}'$ does not extend $B_k$, and $B_{k'}'$ cannot be certified. This contradicts the assumption that $B_{k'}'$ is $x'$-strong committed, proving the theorem.
\end{proof}

\begin{theorem}[Liveness]\label{thm:liveness}
    After GST, if all $c\leq f$ faults are benign (crash) 
    and rounds $r$ to $r+2$ have honest leaders, then the block $B$ proposed by the round $r$ leader is $(2f-c)$-strong committed within at most $n+2$ rounds.
\end{theorem}
\begin{proof}
    After GST the system becomes synchronous, and every honest leader is able to propose a block extending the previous block proposal. By the liveness proof of the original DiemBFT protocol~\cite{baudet2019state}, the round-$r$ block $B$ will be committed after round $r+2$. 
    Since every replica is honest or benign, there will be no fork in the blockchain under synchrony. Within $n$ rounds, each replica will become the leader once due to the round-robin leader rotation and include its \strongvote in the \strongqc. Therefore, for any block $B$ proposed after GST, within $n$ rounds every replica will become the endorser of $B$, except for those $c$ replicas that may crash. Then, within $n+2$ rounds, each of the block in the $3$-chain will have $n-c=3f+1-c$ endorsers, and thus $B$ is $(2f-c)$-strong committed.
\end{proof}

\subsection{Generalizing \strongvote for Better Liveness}
\label{sec:discussion:strongvote}
As mentioned, the liveness guarantee of the protocol can be improved by embedding more information in \strongvote.
Specifically, we can generalize \strongvote to be $\langle \vote, B, r, \mathcal{I} \rangle$, where $\mathcal{I}$ is a set of {\em intervals of round numbers that \strongvote endorses}, i.e., the \strongvote endorses any $B'$ whose round number $r'\in I$.
The previous solution with one \marker corresponds to the special case where $\mathcal{I}=[\marker+1, r]$ where $r$ is the round number of $B$.

The rule for computing $\mathcal{I}$ is also simple. 
When replica $h$ votes for a block $B$ on a chain $\fc$, for every fork $\ff$ on which $h$ ever voted for some conflicting block $B'$, it computes an interval of round numbers $D_{\ff}$ it {\em does not} endorse. 
Then $\mathcal{I}$ is simply computed as $\mathcal{I}=[1, r]\setminus (\cup_{\ff}D_{\ff})$.
The interval $D_{\ff}$ is calculated as $D_{\ff}=[r_l+1, r_h]$, where $r_h$ is the largest round number of any conflicting block $B'$ that $h$ voted for on $\ff$, and $r_l$ is the largest round number of the common ancestor block of both $B$ and $B'$.

It is straightforward that specifying a subset of $\mathcal{I}$ in the \strongvote only affects the liveness of strong commit but not safety.
Therefore, there exists a tradeoff between the amount of information in the \strongvote and the liveness guarantee of the strong commit.
For instance, the solution we presented in Section \ref{sec:diembft:oft} attaches the minimum amount of information in the \strongvote by setting $\mathcal{I}=[\marker+1, r]$ where $r$ is the round number of the \strongvote. As a result, the liveness guarantee there is relatively weak -- strong commits are ensured only with benign faults (Theorem~\ref{thm:liveness}).
On the contrary, if \strongvote attaches the set of intervals for the last $n$ rounds, i.e., $\mathcal{I}=[r-n, r]\setminus (\cup_{\ff}D_{\ff})$, then the liveness guarantee for SFT-DiemBFT can be strengthened as Theorem~\ref{thm:newliveness} (proofs in Appendix~\ref{sec:proofs}).
This is still efficient for moderate values of $n$ since at most $t$ intervals need to be specified during periods of synchrony where $t$ is the actual number of Byzantine faults.

\begin{theorem}\label{thm:newliveness}
    After GST, if the actual number of Byzantine faults is $t\leq f$ and rounds $r$ to $r+5$ have honest leaders, then the block $B$ proposed by the round $r$ leader is $(2f-t)$-strong committed within at most $n+2$ rounds.
\end{theorem}

\section{Implementation and Evaluation}
\label{sec:evaluation}

\paragraph{Implementation details.}
We implement SFT-DiemBFT presented in Section~\ref{sec:diembft:oft} on top of the open source implementation of DiemBFT~\cite{diem,baudet2019state}. 
DiemBFT is written in Rust and has about 400K lines of code.
Our implementation of SFT-DiemBFT has approximately 1.5K lines of code and consists of three main components:
computing \marker, adopting \strongvote and \strongqc, and identifying strong commits. 
To compute the value of \marker, each replica locally stores the highest-round block it has ever voted, for every fork in the blockchain.
When a replica votes, it includes the \marker in the \strongvote, and the leader forms a \strongqc after receiving $2f+1$ \strongvote messages.
Then, whenever a replica receives a new \strongqc, it updates the number of endorsers for all ancestor blocks and their strong commit resilience levels.

\paragraph{Experimental setup.}
We use Amazon EC2 to conduct all the experiments, and execute each replica on a \texttt{c5d.2xlarge} instance. 
Each \texttt{c5d.2xlarge} instance has $8$ vCPUs supported by Intel Xeon Platinum 8000 series processors with a sustained all core Turbo CPU clock speed of up to 3.6 GHz. Each instance's memory size is $16$GiB, and the storage uses 200 GiB NVMe SSD.

We run $n=100$ replica instances in all our experiments, and hence denote $f=(n-1)/3=33$.
Each proposed block contains roughly $1000$ transactions, and has a size of around $450$KB.
Sufficiently many transactions are generated and submitted by the clients so that any leader always has enough transactions to include in its proposed block.

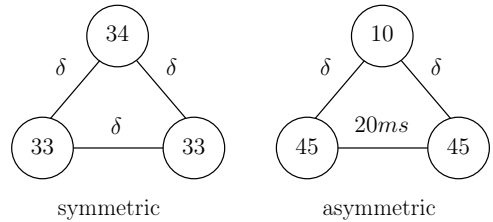
\begin{wrapfigure}{r}{0.4\textwidth}
  \begin{center}

\resizebox{.4\textwidth}{!}{

\tikzset{every picture/.style={line width=0.75pt}} 

\begin{tikzpicture}[x=0.75pt,y=0.75pt,yscale=-1,xscale=1]

\draw   (120,65) .. controls (120,51.19) and (131.19,40) .. (145,40) .. controls (158.81,40) and (170,51.19) .. (170,65) .. controls (170,78.81) and (158.81,90) .. (145,90) .. controls (131.19,90) and (120,78.81) .. (120,65) -- cycle ;
\draw   (59,156) .. controls (59,142.19) and (70.19,131) .. (84,131) .. controls (97.81,131) and (109,142.19) .. (109,156) .. controls (109,169.81) and (97.81,181) .. (84,181) .. controls (70.19,181) and (59,169.81) .. (59,156) -- cycle ;
\draw   (182,156) .. controls (182,142.19) and (193.19,131) .. (207,131) .. controls (220.81,131) and (232,142.19) .. (232,156) .. controls (232,169.81) and (220.81,181) .. (207,181) .. controls (193.19,181) and (182,169.81) .. (182,156) -- cycle ;
\draw    (90.5,132) -- (130.5,85) ;
\draw    (200.5,132) -- (160.5,85) ;
\draw    (109,156) -- (182,156) ;
\draw   (335,65) .. controls (335,51.19) and (346.19,40) .. (360,40) .. controls (373.81,40) and (385,51.19) .. (385,65) .. controls (385,78.81) and (373.81,90) .. (360,90) .. controls (346.19,90) and (335,78.81) .. (335,65) -- cycle ;
\draw   (274,156) .. controls (274,142.19) and (285.19,131) .. (299,131) .. controls (312.81,131) and (324,142.19) .. (324,156) .. controls (324,169.81) and (312.81,181) .. (299,181) .. controls (285.19,181) and (274,169.81) .. (274,156) -- cycle ;
\draw   (397,156) .. controls (397,142.19) and (408.19,131) .. (422,131) .. controls (435.81,131) and (447,142.19) .. (447,156) .. controls (447,169.81) and (435.81,181) .. (422,181) .. controls (408.19,181) and (397,169.81) .. (397,156) -- cycle ;
\draw    (305.5,132) -- (345.5,85) ;
\draw    (415.5,132) -- (375.5,85) ;
\draw    (324,156) -- (397,156) ;

\draw (134,55) node [anchor=north west][inner sep=0.75pt]  [font=\Large] [align=left] {$\displaystyle 34$};
\draw (73,147) node [anchor=north west][inner sep=0.75pt]  [font=\Large] [align=left] {$\displaystyle 33$};
\draw (196,147) node [anchor=north west][inner sep=0.75pt]  [font=\Large] [align=left] {$\displaystyle 33$};
\draw (138,130) node [anchor=north west][inner sep=0.75pt]  [font=\Large] [align=left] {$\displaystyle \delta $};
\draw (93,82) node [anchor=north west][inner sep=0.75pt]  [font=\Large] [align=left] {$\displaystyle \delta $};
\draw (183,82) node [anchor=north west][inner sep=0.75pt]  [font=\Large] [align=left] {$\displaystyle \delta $};
\draw (95,196) node [anchor=north west][inner sep=0.75pt]  [font=\Large] [align=left] {symmetric};
\draw (349,55) node [anchor=north west][inner sep=0.75pt]  [font=\Large] [align=left] {$\displaystyle 10$};
\draw (288,147) node [anchor=north west][inner sep=0.75pt]  [font=\Large] [align=left] {$\displaystyle 45$};
\draw (411,147) node [anchor=north west][inner sep=0.75pt]  [font=\Large] [align=left] {$\displaystyle 45$};
\draw (336,129) node [anchor=north west][inner sep=0.75pt]  [font=\Large] [align=left] {$\displaystyle 20ms$};
\draw (308,82) node [anchor=north west][inner sep=0.75pt]  [font=\Large] [align=left] {$\displaystyle \delta $};
\draw (398,82) node [anchor=north west][inner sep=0.75pt]  [font=\Large] [align=left] {$\displaystyle \delta $};
\draw (310,196) node [anchor=north west][inner sep=0.75pt]  [font=\Large] [align=left] {asymmetric};

\end{tikzpicture}

}
  \end{center}
  \caption{Illustrations of Symmetric and Asymmetric Settings. All $100$ replicas are divided into $3$ regions, with message delays shown on the edges.}
  \label{fig:geo}
\end{wrapfigure}

Each experiment is run for at least $5$ minutes, and each data point is the average value measured over all blocks over all replicas.
Since \strongvote adds very small overhead (one integer) to message size, as expected, we found that the throughput of SFT-DiemBFT is almost identical to that of the original DiemBFT protocol in all our experiments. Hence we omit the throughput results.
We focus on {\em latency} of strong commits of different resilience levels, measured by the time duration from when a block is created to when the block is strong committed.

We evaluate the latency of strong commits in a geo-distributed setting. We consider two scenarios with symmetric and asymmetric network delay distributions, as illustrated in Figure~\ref{fig:geo}. 
In the symmetric setting,  we partition all $100$ replicas evenly into $3$ regions and inject a fixed delay $\delta$ between any pair of replicas across different regions.
In the asymmetric setting, we consider three regions $A,B,C$ with $45,45,10$ nodes respectively. The delay across $A$ and $B$ is $20$ms, while the delays across $C,A$ and $C,B$ are $\delta$. 
We evaluate strong commit latencies with $\delta=100$ms and $200$ms in Section~\ref{sec:evaluation:geo}, and $\delta=100$ms in Section~\ref{sec:evaluation:tradeoff}.

\subsection{Strong Commit Latency}
\label{sec:evaluation:geo}

We first evaluate the symmetric setting. The strong commit latency results are given in Figure~\ref{fig:geo:1}.
As can be seen from the figure, the $x$-strong commit latency increases with $x$ almost linearly, except for $1.1f$-strong commit and $2f$-strong commit that have larger latency increases than other strong commits.
To $1.1f$-strong commit a block, the protocol needs to wait for at least one more round-trip to obtain a new \strongqc, and hence the latency of $1.1f$-strong commits is $4$ round-trips delay ($3$ round-trips by $3$-chain rule plus the additional round-trip for the new \strongqc), while $f$-strong commits needs $3$ round-trips by the $3$-chain rule. 
After that, the latency of strong commits increases slowly as $x$ increase, which implies the \strongqcs in the chain have pretty good diversity and thus each block get endorsed by many endorsers very quickly.
To achieve a higher level of resilience such as the $2f$-strong commit, however, \strongvote from most of the replicas need to be included in some \strongqc  in the chain, including the ``stragglers'' in the system who are out-of-sync due to slow network/computation.
A straggler may have one only chance every $n$ rounds to include its \strongvotes in some \strongqc on the chain, that is, when it becomes the leader. Therefore, the latency for close to $2f$-strong commits are significantly higher.

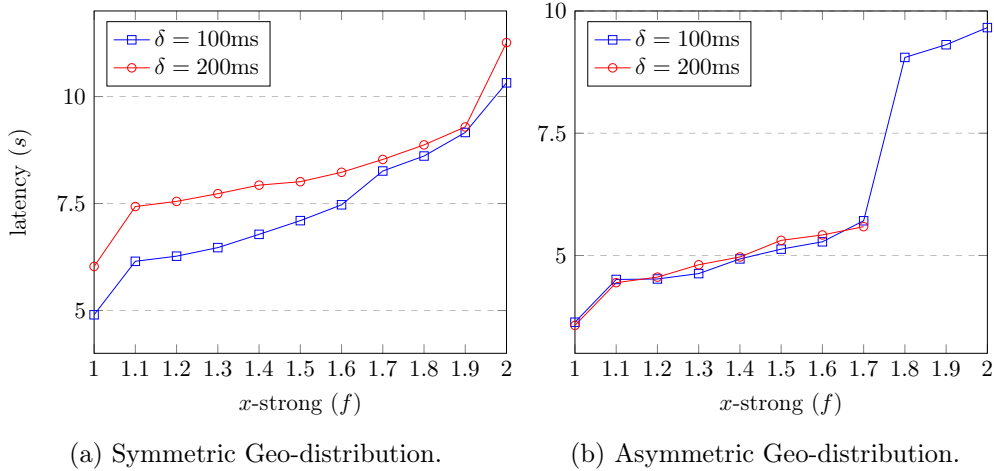
\begin{figure*}[t]
   \begin{minipage}{\textwidth}
     \centering
\subcaptionbox{Symmetric Geo-distribution.\label{fig:geo:1}}{
    \begin{tikzpicture}[scale=0.8]
    \centering
\begin{axis}[
    xlabel={$x$-strong ($f$)},
    ylabel={latency ($s$)},
    xmin=1, xmax=2,
    ymin=4, ymax=12,
    xtick={1,1.1, 1.2,1.3,1.4,1.5,1.6,1.7,1.8,1.9,2},
    ytick={0,5,7.5,10,15,20,25,30,35},
    legend pos=north west,
    ymajorgrids=true,
    grid style=dashed,
]

\addplot[
    color=blue,
    mark=square,
    ]
    coordinates {
    (1,4.9)(1.1,6.15)(1.2,6.27)(1.3,6.47)(1.4,6.78)(1.5,7.1)(1.6,7.47)(1.7,8.26)(1.8,8.61)(1.9,9.16)(2,10.32)
    };
    \addlegendentry{$\delta=100$ms}
    
\addplot[
    color=red,
    mark=o,
    ]
    coordinates {
    (1,6.03)(1.1,7.43)(1.2,7.55)(1.3,7.73)(1.4,7.93)(1.5,8.01)(1.6,8.23)(1.7,8.53)(1.8,8.87)(1.9,9.29)(2,11.26)
    };
    \addlegendentry{$\delta=200$ms}

\end{axis}
\end{tikzpicture}
}%
\subcaptionbox{Asymmetric Geo-distribution.\label{fig:geo:2}}{%
\begin{tikzpicture}[scale=0.8]
\begin{axis}[
    xlabel={$x$-strong ($f$)},
    xmin=1, xmax=2,
    ymin=3, ymax=10,
    xtick={1,1.1, 1.2,1.3,1.4,1.5,1.6,1.7,1.8,1.9,2},
    ytick={0,5,7.5,10,15,20,25,30,35},
    legend pos=north west,
    ymajorgrids=true,
    grid style=dashed,
]

\addplot[
    color=blue,
    mark=square,
    ]
    coordinates {
    (1,3.636)(1.1,4.51)(1.2,4.52)(1.3,4.63)(1.4,4.93)(1.5,5.13)(1.6,5.28)(1.7,5.71)(1.8,9.05)(1.9,9.31)(2,9.66)
    };
    \addlegendentry{$\delta=100$ms}
    
\addplot[
    color=red,
    mark=o,
    ]
    coordinates {
    (1,3.571)(1.1,4.444)(1.2,4.558)(1.3,4.81)(1.4,4.97)(1.5,5.31)(1.6,5.42)(1.7,5.59)
    };
    \addlegendentry{$\delta=200$ms}
\end{axis}
\end{tikzpicture}
}

\caption{Strong Commit Latencies under Geo-distributed Settings.}
   \end{minipage}
\end{figure*}

The latency results of strong commits in the \emph{asymmetric} setting are given in Figure~\ref{fig:geo:2}.
Note that in the asymmetric case, we set the size of the region $C$ to be small, and the delay between $C$ and $A$ or $B$ is large. Therefore, when the leader is from regions $A,B$, the \strongqc will contain only \strongvotes from replicas in $A,B$ because their \strongvote messages always arrive faster.
Moreover, blocks can be $x$-strong committed for $x\leq 1.7f$ with endorsers only from regions $A,B$, and thus the latency of these strong commits is relatively small.
However, to achieve a resilience higher than $1.7f$, the \strongvotes from replicas in the region $C$ also need to be included in some \strongqcs on the chain.
When $\delta=100$ms, those \strongvotes can be included only when the replicas in the region $C$ become the leader, which is $10$ rounds out of $100$ rounds due to the round-robin leader rotations. Since such \strongqcs are less frequent and have larger round latencies,  the latencies to achieve $x$-strong commits for $x\geq 1.8f$ is significantly higher than others.
When $\delta=200$ms, we discover that any leader from the region $C$ will timeout and be replaced. 
As a result, any \strongqc in the blockchain never contains \strongvotes from replicas in $C$, and thus the highest strong commit resilience level is $2f-10=1.7f$.
In practice, there may indeed exist such ``outcast replicas'' in the system whose \strongvotes are never included in the blockchain, and therefore the highest level of strong commit may not be achieved. However, these replicas in the systems should be reconfigured or replaced, since they are not actively participating in the protocol.

\subsection{Tradeoff between regular commit and strong commit latencies}
\label{sec:evaluation:tradeoff}

From the previous discussions, we know that bottleneck of strong commit latency is due to stragglers in the system, whose \strongvotes are rarely included in \strongqcs in the blockchain.
Therefore, a natural idea to speed up the strong commits is to intentionally wait for those stragglers, by asking leaders to wait for an extra period of time after receiving $2f+1$ \strongvotes and including any new \strongvotes in the \strongqc to improve QC diversity. 
This technique improves strong commit latency but sacrifices round latency, which leads to a higher latency for regular commits.
We explore this tradeoff space and characterize the results in Figure~\ref{fig:tradeoff}.

We use the symmetric geo-distributed setting with $\delta=100$ms.
For each round, we gradually increase the round latency such that the leader can include more than $2f+1$ \strongvotes in the \strongqc. 
As can be seen in Figure~\ref{fig:tradeoff}, a relatively small increase in the regular commit latency (about $0.5$s) can significantly reduce the $2f$-strong commit latency from about $10$s to about $5$s. 
For strong commits with resilience other than $2f$, the commit latency first decrease dramatically and then slowly increase. 
The straight line at the bottom represents the latency for the regular commit. 
All strong commit latency curves coincide with the regular commit latency curve after some threshold.
The reason is that when each leader includes $Q\geq 2f+1$ \strongvotes in the \strongqcs, a $3$-chain for a regular commit ($f$-strong) happen at the same time as a $(Q-f-1)$-strong commit; hence the two curves merge.

\begingroup
\setlength{\intextsep}{0pt}

\begin{wrapfigure}{r}{0.47\textwidth}
\centering

\resizebox{.47\textwidth}{!}{

\begin{tikzpicture}
\begin{axis}[
    xlabel={regular commit latency ($s$)},
    ylabel={strong commit latency ($s$)},
    xmin=4.7, xmax=5.4,
    ymin=4, ymax=11,
    xtick={4,4.5,4.75,5,5.25,5.5,6},
    ytick={0,5,7.5,10,15,20,25,30,35},
    legend pos=north east,
    ymajorgrids=true,
    grid style=dashed,
]

\addplot[
    color=blue,
    mark=square,
    ]
    coordinates {
    (4.7, 10.32)(4.73, 9.78)(4.77,9.72)(4.79, 8.68)(4.88, 8.14)(4.94, 7.49)(5.11, 7.14)(5.36, 5.36)
    };
    \addlegendentry{$2f$-strong commit}
    
\addplot[
    color=black,
    mark=o,
    ]
    coordinates {
    (4.7, 8.91)(4.73, 7.33)(4.77,6.86)(4.79, 6.26)(4.88, 6.16)(4.94, 4.94)(5.11, 5.11)(5.36, 5.36)
    };
    \addlegendentry{$1.8f$-strong commit}
    
\addplot[
    color=red,
    mark=triangle,
    ]
    coordinates {
    (4.7, 7.47)(4.73, 6.26)(4.77,6.22)(4.79, 4.79)(4.88, 4.88)(4.94, 4.94)(5.11, 5.11)(5.36, 5.36)
    };
    \addlegendentry{$1.6f$-strong commit}

\addplot[
    color=orange,
    mark=x,
    ]
    coordinates {
    (4.7, 6.78)(4.73, 6.15)(4.77,4.77)(4.79, 4.79)(4.88, 4.88)(4.94, 4.94)(5.11, 5.11)(5.36, 5.36)
    };
    \addlegendentry{$1.4f$-strong commit}
    
\addplot[
    color=green,
    mark=+,
    ]
    coordinates {
    (4.7, 6.15)(4.73, 4.73)(4.77,4.77)(4.79, 4.79)(4.88, 4.88)(4.94, 4.94)(5.11, 5.11)(5.36, 5.36)
    };
    \addlegendentry{$1.2f$-strong commit}
    
\end{axis}
\end{tikzpicture}

}
    \caption{Strong/regular Commit Latency Tradeoffs.}
    \label{fig:tradeoff}
\end{wrapfigure}
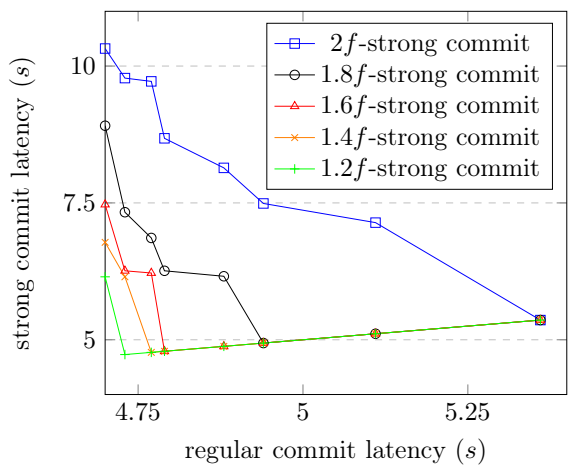

Given the tradeoff observed in Figure~\ref{fig:tradeoff}, we see some interesting strategies for practice. First, a short extra wait seems worthwhile as it sacrifices regular commit latency slightly but significantly reduces strong commit latency. 
Second, the tradeoff can be made in a dynamic fashion.
For instance, if replicas want to strong commit a certain block $B$ quickly (say it contains many high-valued transactions), the leader who proposes $B$ and the next few leaders can increase their round latency to include more \strongvotes in their \strongqcs. 
This way, only block $B$ and a few subsequent blocks incur higher regular commit latencies. Once block $B$ is strong committed with sufficient strength, future leaders can change back their round latencies.

\section{Discussions and Extensions}
\label{sec:discussion}


\paragraph{Proving Strong Commit to Light Clients.}
\label{sec:lightclient}
One important functionality in practice is to prove the strong commit to a light client, which models lightweight applications interacting with the blockchain such as wallet apps, who have no access to the blockchain or do not store the blockchain locally.
To prove the strong commit efficiently, the protocol can include an additional \texttt{Log} on every block proposal, which records any update on the strong commit level of previous blocks due to the new \strongqc contained in the proposal.
Once the block proposal is certified ($2f+1$ replicas voted), at least one honest replica agrees on the strong commit update assuming the number of Byzantine faults does not exceed $2f$ (recall our SFT solutions provides up to $2f$ resilience). Hence it is sufficient to show the certified \texttt{Log} to the light client to prove the strong commit.

\paragraph{Optimizations for Strong Commit Latencies.}
The latency of achieving $x$-strong commit for a block depends on how fast the endorsers appear in the chain extension.
In practice, how to optimize the latency of strong commits under various real-world scenarios is an interesting problem.
For instance, to ensure the diversity of \strongqc, the protocol designer may intentionally extend the duration of a certain round for the leader to collect more votes once every $n$ rounds. 
Moreover, the QC diversity requirement implied by strong commit is closely aligned with the task of monitoring the health conditions of the replicas in the system, which can be done via observing the QCs in the chain and detecting slow replicas.
We leave other possible practical optimizations as a future work.

\paragraph{Conflicting Transactions.}
With the notion of strong commits, transactions with lower resilience requirements from later blocks might be committed before transactions with higher resilience requirements from earlier blocks on the chain. 
For example, $txn_1$ is contained in $B_1$ and requires a $2f$-strong commit, and $txn_2$ is contained in $B_2$ extending $B_1$ but only requires a $f$-strong commit. Then it is possible that $txn_2$ is $f$-strong committed before $txn_1$ is $2f$-strong committed, and there would be an issue if $txn_1$ and $txn_2$ are conflicting transactions.
To avoid such conflict, the protocol can ask the leader to propose conflicting transactions only after the block containing the earlier transaction is already strong committed.
Intuitively, this means when a high-valued transaction is still waiting to be strong committed, other transactions issued by the same sender/receiver should not be processed until the earlier transaction is committed.
We leave other possible mechanisms for resolving conflicting transactions as a future work.

\section{Related Works}
\label{sec:relatedwork}
\paragraph{Byzantine fault tolerant replication.}
Byzantine fault tolerant state machine replication has been extensively studied in the literature. 
Under partial synchrony, PBFT is the first practical leader-based BFT SMR protocol with $O(n^2)$ communication cost per decision under a stable leader and $O(n^3)$ communication cost to replace a leader. 
PBFT has been deployed in many systems such as BFT-SMART~\cite{bessani2014state} and Hyperledger fabric~\cite{androulaki2018hyperledger}.
Following PBFT, an important line of work targets at reducing the communication cost of BFT SMR protocols, including Zyzzyva~\cite{kotla2007zyzzyva,abraham2018revisiting},
Terdermint~\cite{buchman2016tendermint}, SBFT~\cite{gueta2019sbft},
Casper~\cite{buterin2017casper} and HotStuff~\cite{yin2019hotstuff}.
The last two works above developed the framework of chain-based BFT .
There are also recent works on chain-based BFT SMR under synchrony~\cite{hanke2018dfinity,Abraham2020SyncHS,abraham2020optimality}, including Dfinity~\cite{hanke2018dfinity}, Sync HotStuff~\cite{Abraham2020SyncHS}, SMR with optimal optimistic responsiveness~\cite{abraham2020optimality}, and $1\Delta$-SMR~\cite{abraham2020optimal}. The above synchronous BFT SMR protocols proceed in a non-conventional non-lockstep fashion, and aim to improve the commit latency under an honest leader.
For asynchronous systems where the message delay is unbounded, it is impossible to solve consensus deterministically in the presents of faults~\cite{fischer1985impossibility}. Several recent proposals on asynchronous BFT replication protocols focuses on improving the communication complexity and latency, including HoneyBadgerBFT~\cite{miller2016honey}, VABA~\cite{abraham2019asymptotically} and Aleph~\cite{gkagol2019aleph}.

\paragraph{Flexibility in BFT protocols.}
The most related and inspiring work for this paper is Flexible Byzantine Fault Tolerance (FBFT)~\cite{malkhi2019flexible}, where one BFT SMR solution support clients with various fault and synchronicity beliefs. 
In FBFT, each client can specify its assumption on the fault threshold in terms of Byzantine and Alive-but-corrupt (a-b-c) faults, where the a-b-c fault only attacks safety but not liveness, as well as the timing model including synchrony and partial synchrony. 
The safety and liveness guarantee of the protocol only hold for clients with correct fault and timing assumptions. 
For partial synchrony, Flexible BFT extends PBFT~\cite{castro1999practical} with flexible Byzantine quorums, where several different types of quorums used in the protocol are identified, and the optimal resilience solution is achieved with two quorum sizes $q_r,q_c$, where the replicas proceed with quorum size $q_r$ and clients commit with quorum size $q_c$. 
Intuitively, our optimistic fault tolerance solution shows how to increase the quorum size of $q_c$ with indirect vote messages in the chain extension, in contrast to Flexible BFT where $q_c$ votes come from direct votes. 
The FBFT algorithm can be adopted to achieve SFT. However, doing so would incur $O(n^2)$ message complexity per block decision, while our SFT-DiemBFT keeps the complexity linear.
A related work~\cite{li2007beyond} guarantees a weaker consistency called fork consistency in PBFT when the number of Byzantine faults exceeds $f$. In contrast, our SFT can guarantee safety (linearizability) for any $x$-strong committed blocks as long as the number of faults does not exceed $x$. 
A recent work called BFT Forensics~\cite{sheng2020bft} shows how to detect at least $f+1$ Byzantine parties under safety violations due to the total number of Byzantine faults $>f$. In contrast, our SFT aims to guarantee safety even if the number of Byzantine faults exceeds the assumed threshold, by utilizing the optimistic periods to strong commit the blocks.
Another related line of work investigates solving BFT replication under synchrony with an asynchronous fallback, with optimal resilience tradeoffs between synchrony and asynchrony~\cite{blum2019synchronous,blum2020network} or optimal communication complexity~\cite{spiegelman2020search}.

\paragraph{Optimistic BFT protocols.}
Another related line of work on optimistic BFT protocols focuses on including a fast path in BFT protocols to improve the performance such as communication complexity and latency under the optimistic scenarios, for asynchronous protocols~\cite{kursawe2005optimistic, kursawe2002optimistic}, partially synchronous protocols~\cite{abd2005fault, kotla2007zyzzyva, guerraoui2010next} and synchronous protocols~\cite{pass2018thunderella, cryptoeprint:2018:980, nibesh2020optimality, cryptoeprint:2020:406}.
It should be noted that the goal of this paper is different from the works on optimistic BFT. Our strengthened fault tolerance utilizes the optimistic periods to tolerate more than one-third of the faults even after the optimistic periods, while the literature on optimistic BFT provides a fast path for efficient decisions during the optimistic periods. 

\paragraph{Asymmetric Byzantine Quorum System.}
A related line of work studies general Byzantine quorum system~\cite{malkhi1998byzantine} where a failure-prone set characterizes all sets of possibly faulty replicas instead of just a threshold number. 
Several related works on asymmetric Byzantine quorum system~\cite{cachin2020asymmetric,cachin2020asymmetricconsensus} investigate a model where replicas may have different beliefs on the failure-prone sets, and how to achieve agreement in such a model. Comparing with FBFT and our work that provide the flexibility of various resilience guarantees from the clients' perspective optimistically, the asymmetric Byzantine quorum system provides the flexibility at the replica level for worst-case scenarios. 
As a result, the guarantees of the protocol may not hold for all honest replicas in the asymmetric Byzantine quorum system~\cite{cachin2020asymmetric,cachin2020asymmetricconsensus}.
Another related line of work investigates federated Byzantine quorum systems (FBQS) where every replica has its own set of trusted nodes, known as Stellar consensus~\cite{lokhava2019fast,losa2019stellar,garcia2020deconstructing}. 
More details on the difference between asymmetric Byzantine quorum systems and federated Byzantine quorum systems can be found in~\cite{cachin2020asymmetric}.

\section{Conclusion}
\label{sec:conclusion}
In this paper, we propose solutions for efficiently implementing strengthened fault tolerance in chain-based Byzantine fault tolerant state machine replication protocols under partial synchrony, where blocks in the chain can gradually obtain higher assurance against more than one-third corruption during the optimistic period.
We also implement the proposed solution and evaluate the performance under real-world scenarios.

\bibliographystyle{plain}
\bibliography{references}

\appendix

\section{Missing Proofs}
\label{sec:proofs}

\subsection{Proof for Solution with Generalized \strongvote}
\label{sec:proofs:strongvote}
In this section, we prove the safety of  the extension in Section~\ref{sec:discussion:strongvote}.
First, in the proof of Lemma~\ref{lem:1} with generalized \strongvote, we have the same observation that any honest replica voted for a conflicting block $B'$ will not be the endorser of $B$, because $r\in D_{\ff}$ and thus $r\notin \mathcal{I}$. Rest of the proof of Lemma~\ref{lem:1} follows.
Similarly, we can also show that Lemma~\ref{lem:2} holds with generalized \strongvote as well.
In the proof of Lemma~\ref{lem:2}, the only difference with generalized \strongvote is the case when $r'\geq r+3$ where $r'$ is the round number of the conflicting certified block $B'$.
Similarly, any honest replica voted for $B'$ will not be the endorser of $B_{r+2}$, since $r+2\in [r,r']\subseteq D_{\ff}$ and thus $r+2\notin \mathcal{I}$. Rest of the proof follows, and Theorem~\ref{thm:safety} also holds.

\subsection{Proof of Theorem~\ref{thm:newliveness}}
\begin{proof}
    By the liveness proof of the original DiemBFT protocol~\cite{baudet2019state}, the round-$r$ block $B$ will be committed in round $r+3$. 
    We first prove that the round-$r$ block will get $n-t$ endorsers within $n$ rounds.
    Since the system is synchronous, all honest replicas will only vote for blocks with round number $\geq r$ after receiving the round-$r+2$ block due to the locking and voting rule. Also, leaders of round $r$ to $r+2$ are honest, no fork is created during round $r$ to $r+2$, which implies that all honest replicas have $r\notin D_{\ff}$ since they will not vote for any block extending round $<r$ after round $r+2$.
    Therefore, every later vote from the honest replica endorses $B$. Within $n$ rounds, each replica will become the leader once due to the round-robin leader rotation and include its \strongvote in the \strongqc. Therefore, within $n$ rounds every replica will become the endorser of $B$, except for those $t$ Byzantine replicas that may not vote. 
    Similarly, for the round-$(r+1)$ and round-$(r+2)$ blocks, within $n+2$ rounds, each of them will get $n-t$ endorsers. By the strong commit rule, the round-$r$ block $B$ is $(2f-t)$-strong committed.
\end{proof}

\section{Adapting FBFT~\cite{malkhi2019flexible} for DiemBFT}
\label{sec:fbft}
FBFT proposes the notion of flexible BFT quorums, which identifies quorums of different sizes to provide different levels of fault tolerance guarantee for consensus decisions. Although in FBFT~\cite{malkhi2019flexible} only the solution based on PBFT is proposed, it is straightforward to extend the idea of flexible BFT quorums to chain-based BFT SMR protocols such as HotStuff/DiemBFT as well.
The strong commit rule of FBFT on DiemBFT is the following.

{\bf Strong commit rule.}
    The replica $x$-strong commits a block $B_k$ and all its ancestors, if and only if there exists three adjacent blocks $B_k,B_{k+1},B_{k+2}$ in the chain with consecutive round numbers, 
    and each of $B_k,B_{k+1},B_{k+2}$ has \underline{at least $x+f+1$ distinct signed \vote messages}.
    
Intuitively, the higher level of resilience is achieved by directly increasing the quorum certificate size in the $3$-chain rule when strong committing a block. 
However, the protocol loses liveness with QC size larger than $2f+1$, since the number of Byzantine faults may be $f$. 
Therefore, replicas should advance to the next round with the standard quorum size of $2f+1$, and any extra \vote messages from earlier rounds need to be multicasted by the leader of that round.
Due to the extra work the leader needs to conduct for every round, the amortized message complexity per round increases in the FBFT approach. 
Since there may be up to $f$ extra votes not included in the QC, the leader may need to multicast up to $f$ extra messages, leading to an amortized message complexity $O(fn)=O(n^2)$ per round and thus $O(n^2)$ per block decision.

\section{A Counter-example for Counting All Indirect Votes}
\label{sec:diembft:quorumdiversity}
\begin{figure}
    \centering

\resizebox{\textwidth}{!}{

\tikzset{every picture/.style={line width=0.75pt}} 

\begin{tikzpicture}[x=0.75pt,y=0.75pt,yscale=-1,xscale=1]

\draw   (20,150.67) -- (60.2,150.67) -- (60.2,189.6) -- (20,189.6) -- cycle ;
\draw   (110,110.67) -- (150.2,110.67) -- (150.2,149.6) -- (110,149.6) -- cycle ;
\draw   (190,110.67) -- (230.2,110.67) -- (230.2,149.6) -- (190,149.6) -- cycle ;
\draw   (270,110.67) -- (310.2,110.67) -- (310.2,149.6) -- (270,149.6) -- cycle ;
\draw   (350,110.67) -- (390.2,110.67) -- (390.2,149.6) -- (350,149.6) -- cycle ;
\draw   (110,190.67) -- (150.2,190.67) -- (150.2,229.6) -- (110,229.6) -- cycle ;
\draw   (190,190.67) -- (230.2,190.67) -- (230.2,229.6) -- (190,229.6) -- cycle ;
\draw  [fill={rgb, 255:red, 0; green, 0; blue, 0 }  ,fill opacity=1 ] (170,123) -- (175,130) -- (170,137) -- (165,130) -- cycle ;
\draw    (150.5,130) -- (190.5,130) ;
\draw  [fill={rgb, 255:red, 0; green, 0; blue, 0 }  ,fill opacity=1 ] (250,123) -- (255,130) -- (250,137) -- (245,130) -- cycle ;
\draw    (230.5,130) -- (270.5,130) ;
\draw  [fill={rgb, 255:red, 0; green, 0; blue, 0 }  ,fill opacity=1 ] (330,123) -- (335,130) -- (330,137) -- (325,130) -- cycle ;
\draw    (310.5,130) -- (349.5,130) ;
\draw  [fill={rgb, 255:red, 0; green, 0; blue, 0 }  ,fill opacity=1 ] (170,203) -- (175,210) -- (170,217) -- (165,210) -- cycle ;
\draw    (150.5,210) -- (190.5,210) ;
\draw  [fill={rgb, 255:red, 0; green, 0; blue, 0 }  ,fill opacity=1 ] (82,163) -- (87,170) -- (82,177) -- (77,170) -- cycle ;
\draw    (60.5,170) -- (86,170) ;
\draw    (109.5,130) -- (87,170) ;
\draw    (87,170) -- (109.5,211) ;
\draw   (270,190.67) -- (310.2,190.67) -- (310.2,229.6) -- (270,229.6) -- cycle ;
\draw  [fill={rgb, 255:red, 0; green, 0; blue, 0 }  ,fill opacity=1 ] (250,203) -- (255,210) -- (250,217) -- (245,210) -- cycle ;
\draw    (230.5,210) -- (270.5,210) ;
\draw   (350,190.67) -- (390.2,190.67) -- (390.2,229.6) -- (350,229.6) -- cycle ;
\draw  [fill={rgb, 255:red, 0; green, 0; blue, 0 }  ,fill opacity=1 ] (330,203) -- (335,210) -- (330,217) -- (325,210) -- cycle ;
\draw    (310.5,210) -- (350.5,210) ;
\draw   (430,190.67) -- (470.2,190.67) -- (470.2,229.6) -- (430,229.6) -- cycle ;
\draw  [fill={rgb, 255:red, 0; green, 0; blue, 0 }  ,fill opacity=1 ] (410,203) -- (415,210) -- (410,217) -- (405,210) -- cycle ;
\draw    (390.5,210) -- (430.5,210) ;
\draw    (170.5,117) -- (170.5,89) ;
\draw [shift={(170.5,87)}, rotate = 450] [color={rgb, 255:red, 0; green, 0; blue, 0 }  ][line width=0.75]    (10.93,-3.29) .. controls (6.95,-1.4) and (3.31,-0.3) .. (0,0) .. controls (3.31,0.3) and (6.95,1.4) .. (10.93,3.29)   ;
\draw    (249.5,117) -- (231.6,89.67) ;
\draw [shift={(230.5,88)}, rotate = 416.77] [color={rgb, 255:red, 0; green, 0; blue, 0 }  ][line width=0.75]    (10.93,-3.29) .. controls (6.95,-1.4) and (3.31,-0.3) .. (0,0) .. controls (3.31,0.3) and (6.95,1.4) .. (10.93,3.29)   ;
\draw    (329.5,116) -- (329.5,88) ;
\draw [shift={(329.5,86)}, rotate = 450] [color={rgb, 255:red, 0; green, 0; blue, 0 }  ][line width=0.75]    (10.93,-3.29) .. controls (6.95,-1.4) and (3.31,-0.3) .. (0,0) .. controls (3.31,0.3) and (6.95,1.4) .. (10.93,3.29)   ;
\draw    (170,240) -- (169.54,264) ;
\draw [shift={(169.5,266)}, rotate = 271.1] [color={rgb, 255:red, 0; green, 0; blue, 0 }  ][line width=0.75]    (10.93,-3.29) .. controls (6.95,-1.4) and (3.31,-0.3) .. (0,0) .. controls (3.31,0.3) and (6.95,1.4) .. (10.93,3.29)   ;
\draw    (331,240) -- (330.54,264) ;
\draw [shift={(330.5,266)}, rotate = 271.1] [color={rgb, 255:red, 0; green, 0; blue, 0 }  ][line width=0.75]    (10.93,-3.29) .. controls (6.95,-1.4) and (3.31,-0.3) .. (0,0) .. controls (3.31,0.3) and (6.95,1.4) .. (10.93,3.29)   ;
\draw    (250,240) -- (270.23,264.46) ;
\draw [shift={(271.5,266)}, rotate = 230.41] [color={rgb, 255:red, 0; green, 0; blue, 0 }  ][line width=0.75]    (10.93,-3.29) .. controls (6.95,-1.4) and (3.31,-0.3) .. (0,0) .. controls (3.31,0.3) and (6.95,1.4) .. (10.93,3.29)   ;
\draw    (410,240) -- (390.74,264.43) ;
\draw [shift={(389.5,266)}, rotate = 308.25] [color={rgb, 255:red, 0; green, 0; blue, 0 }  ][line width=0.75]    (10.93,-3.29) .. controls (6.95,-1.4) and (3.31,-0.3) .. (0,0) .. controls (3.31,0.3) and (6.95,1.4) .. (10.93,3.29)   ;

\draw (40.1,170.13) node  [font=\large]  {$B_{r-1}$};
\draw (130.1,130.13) node  [font=\large]  {$B_{r}$};
\draw (210.1,130.13) node  [font=\large]  {$B_{r+1}$};
\draw (290.1,130.13) node  [font=\large]  {$B_{r+2}$};
\draw (370.1,130.13) node  [font=\large]  {$B_{r+3}$};
\draw (130.1,210.13) node  [font=\large]  {$B'_{r+1}$};
\draw (210.1,211.13) node  [font=\large]  {$B'_{r+4}$};
\draw (290.1,210.13) node  [font=\large]  {$B'_{r+5}$};
\draw (370.1,211.13) node  [font=\large]  {$B'_{r+6}$};
\draw (450.1,211.13) node  [font=\large]  {$B'_{r+7}$};
\draw (18,60) node [anchor=north west][inner sep=0.75pt]  [font=\small] [align=left] {$\displaystyle QC_{r} =QC_{r+1} =\{h_{1} ,...,h_{f}\} \cup \{b_{1} ,...,b_{f+1}\}$};
\draw (302,55) node [anchor=north west][inner sep=0.75pt]  [font=\small] [align=left] {$\displaystyle QC_{r+2} =\left\{h_{1} ,...,h_{f} ,\underline{h_{f+1}}\right\} \cup \{b_{1} ,...,b_{f}\}$};
\draw (157,143) node [anchor=north west][inner sep=0.75pt]  [font=\large] [align=left] {$\displaystyle QC_{r}$};
\draw (233,143) node [anchor=north west][inner sep=0.75pt]  [font=\large] [align=left] {$\displaystyle QC_{r+1}$};
\draw (314,143) node [anchor=north west][inner sep=0.75pt]  [font=\large] [align=left] {$\displaystyle QC_{r+2}$};
\draw (153,223) node [anchor=north west][inner sep=0.75pt]  [font=\large] [align=left] {$\displaystyle QC'_{r+1}$};
\draw (16,276) node [anchor=north west][inner sep=0.75pt]  [font=\small] [align=left] {$\displaystyle QC'_{r+1} =\left\{\underline{h_{f+1}} ,...,h_{2f}\right\} \cup \{b_{1} ,...,b_{f+1}\}$};
\draw (267,281) node [anchor=north west][inner sep=0.75pt]  [font=\small] [align=left] {$\displaystyle QC=\{h_{1} ,...,h_{f} ,h_{f+2}\} \cup \{b_{1} ,...,b_{f+1}\}$};
\draw (412,122) node [anchor=north west][inner sep=0.75pt]  [font=\large] [align=left] {$\displaystyle h_{f+1}$ {\small commits} $\displaystyle B_{r}$};
\draw (483,201) node [anchor=north west][inner sep=0.75pt]  [font=\large] [align=left] {$\displaystyle h_{1} ,...,h_{f}$ {\small commits} $\displaystyle B'_{r+4}$};

\end{tikzpicture}
}
    
    \caption{An Counter-example for Counting All Indirect Votes.}
    \label{fig:qc_diveristy}
\end{figure}
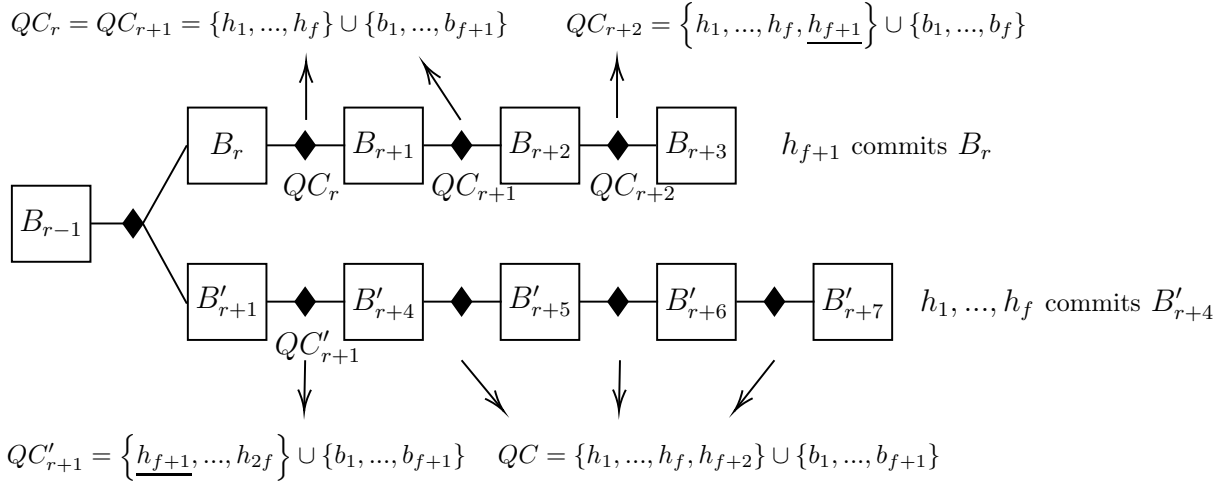

As mentioned in Section~\ref{sec:diembft:oft}, the approach that simply counts all indirect votes in the chain that extends a block $B$ as direct votes for $B$ may overcalculates the fault tolerance of $B$.
We show a simple counter-example to demonstrate why not all indirect votes should be counted, as shown in Figure~\ref{fig:qc_diveristy}.
Suppose that there are $f+1$ Byzantine replicas named $b_1,...,b_{f+1}$, and rest of the $2f$ replicas $h_1,...,h_{2f}$ are honest.
Consider a possible fork in the blockchain as follows.
In round $r$, after the leader proposes $B_r$, $f$ honest replicas $h_1,...,h_f$ and $f+1$ Byzantine replicas $b_1,...,b_{f+1}$ vote for $B_r$. 
In round $r+1$, the leader is Byzantine, it proposes two conflicting blocks $B_{r+1}$ (extending $B_{r}$) and $B'_{r+1}$ (extending an earlier block $B_{r-1}$).
The same set of replicas above vote for $B_{r+1}$, and the remaining $f$ honest replicas together with $f+1$ Byzantine replicas vote for $B_{r+1}'$.
Thus, both $B_{r+1}$ and $B_{r+1}'$ get certified.
In round $r+2$, the leader proposes block $B_{r+2}$ that extends $B_{r+1}$, and $f+1$ honest replicas $h_1,...,h_f,h_{f+1}$ and $f$ Byzantine replicas $b_1,...,b_f$ vote for $B_{r+2}$.
Note that the honest replica $h_{f+1}$ can vote for $B_{r+2}$ after voting for $B'_{r+1}$ due to the voting rule.
Then, due to the new vote of replica $h_{f+1}$, the $3$-chain $B_{r},B_{r+1},B_{r+2}$ now has a QC of size $2f+2$ for each block, which implies that $B_r$ is $(f+1)$-strong committed. By definition, no other conflicting $(f+1)$-strong commit should happen since the actual number of faults is $\leq f+1$.
However, a conflicting $(f+1)$-strong commit can happen, due to the fact that the vote from $h_{f+1}$ is falsely counted as a vote contributing to the resilience of $B_r$.
Continuing the example in Figure~\ref{fig:qc_diveristy}, in round $r+4$  the leader is Byzantine and proposes $B_{r+4}'$ extending $B_{r+1}'$. By the voting rule, any honest replica is able to vote for $B_{r+4}'$ because its $r_{lock}\leq r+1$ due to the locking rule. Similarly, replicas are able to vote for subsequent blocks extending $B_{r+4}'$, and thus $(f+1)$-strong commit the block $B_{r+4}'$.

Therefore, from the counter-example above, we observe that not all indirect votes should be considered a vote for the block when computing the block's resilience.
Our solution SFT-DiemBFT in Section~\ref{sec:diembft:oft} amends this idea by identifying the correct set of indirect votes for computing strong commits.

\section{Streamlet with Strengthened Fault Tolerance}
\label{sec:streamlet}

Another recent proposal on chain-based BFT SMR is Streamlet~\cite{chan2020streamlet}, which focuses on the protocol simplicity.
We show that our approach for implementing strengthened fault tolerance based on endorsements from the indirect vote messages can also apply to the Streamlet protocol.
In Section~\ref{sec:streamlet:overview}, we will first present the overview of the Streamlet protocol, and then show our SFT solution named SFT-Streamlet in Section~\ref{sec:streamlet:oft}.
Comparing to SFT-DiemBFT, one advantage of SFT-Streamlet is that it requires more effort from the adversary to revert a strong commit (see details in Section~\ref{sec:streamlet:comparison}).

\subsection{Overview of Streamlet}
\label{sec:streamlet:overview}
We present the Streamlet protocol by completing the prototype in Section \ref{sec:preliminary:prototype}, as shown in Figure~\ref{fig:streamlet}. The main motivation of Streamlet is to achieve simplicity.
Recall that we say a chain is certified if all its blocks are certified.
The protocol advances in rounds, and each round has a duration $2\Delta$ where $\Delta$ is the assumed maximum network delay.
In each round, the leader proposes a block extending the longest certified chain it knows.
The replicas will vote for the first proposal by the leader, if the proposed block extends one of the longest certified chains it has seen.
When there exist three adjacent certified blocks with consecutive round numbers in the blockchain, the first two blocks in the $3$-chain and all the ancestor blocks are committed. 

Compared with the DiemBFT protocol~\cite{baudet2019state} (Section~\ref{sec:diembft:overview}), Streamlet has the following main differences:
(1) To simplify the protocol, Streamlet adopts a simple Pacemaker that advances the round in a lock-step fashion with an assumed network delay $\Delta$, and thus the protocol does not advance with network speed (not responsive).
(2) Replicas always propose and vote for the longest chain, thus the proposing/voting/locking/commit rules are different. In other words, the rules depend on the {\em height} of the block in the blockchain. We call such approach {\em height-based} as different from HotStuff/DiemBFT which is {\em round-based}.
(3) In each round, the replicas send \vote messages to all other replicas instead of just the next leader. Replicas also forward any unseen message it received to others, leading to a message complexity $O(n^3)$ per round, and hence amortized $O(n^3)$ cost per block decision.
As can be seen above, Streamlet sacrifices performance such as message complexity and commit latency to achieve protocol simplicity.

\begin{figure}[t!]
    \centering
    \begin{mybox}
The Streamlet protocol can be defined by the prototype in Figure~\ref{fig:prototype} with the following components.

For simplicity, the Streamlet also assumes a message echo mechanism, where an honest process forwards a message to all other processes when it receives a previously unseen message.
\begin{itemize}[noitemsep,topsep=0pt]
    \item \underline{\state.}
    Each replica locally keeps all certified blocks it has seen.

    \item \underline{Proposing rule.}
    The leader proposes a height-$k$ block $B_k=(H(B_{k-1}), qc, txn)$ extending the last block $B_{k-1}$ of the {\em longest certified chain} known to the replica, and $qc$ is the quorum certificate of $B_{k-1}$.
    
    \item \underline{Voting rule.}
    Upon receiving the first proposal $\langle \propose, B, r\rangle_{L_r}$ at round $r$,
    multicast a \vote message in the form of $\langle \vote, B, r \rangle_i$ to all other replicas iff
    $B_k$ extends one of the {\em longest certified chains} known to the replica.
    
    \item \underline{Locking rule.}
    Upon receiving a certified block, the replica updates its longest certified chain.
    
    \item \underline{Commit rule.}
    The replica commits a block $B_k$ and all its ancestors, if and only if there exist three adjacent certified blocks $B_{k-1},B_k,B_{k+1}$ in the chain with consecutive round numbers, i.e., 
    $B_{k+1}.round=B_k.round+1=B_{k-1}.round+2$.
    
    \item \underline{Synchronization rule.}
    After $2\Delta$ time in round $r$ where $\Delta$ is the assumed maximum network delay after GST, the replica advances the round number to $r+1$.
    
    \end{itemize}

    \end{mybox}
    \caption{Streamlet Protocol.}
    \label{fig:streamlet}
\end{figure}

\subsection{SFT-Streamlet}
\label{sec:streamlet:oft}

This section presents our SFT-Streamlet protocol for implementing strengthened fault tolerance for Streamlet.
The idea is analogous to SFT-DiemBFT presented in Section~\ref{sec:diembft:oft}.
The main difference is that the \marker in the \strongvote of SFT-Streamlet records the {\em largest height number} of any voted conflicting block, instead of the round number in SFT-DiemBFT.
The definition of endorsement and strong commit rule is also slightly different.
Endorsement now has an additional parameter $k$, and a \strongvote for block $B'$ $k$-endorses a block $B$ if and only if $B=B'$, or $B'$ extends $B$ and $\marker<k$.
The strong commit rule for an $x$-strong commit requires a $3$-chain of blocks $B_{k-1},B_{k},B_{k+1}$ with consecutive round numbers and each has $\geq x+f+1$ $k$-endorsers.

\begin{figure}[t!]
    \centering
    \begin{mybox}
The SFT-Streamlet protocol modifies the original Streamlet protocol in Figure~\ref{fig:diembft} with the following changes.

\begin{itemize}[noitemsep,topsep=0pt]
    \item {\bf Changes to local \state.} 
    For every fork in the blockchain, the replica keeps the highest voted block on that fork.
    
    \item {\bf Strong-vote and strong-QC.}
    For a \vote message of some block $B$, the replica additionally includes a $\marker=\max\{k ~|~ \text{$B'_k$ conflicts $B$ and replica voted for $B'_k$}\}$ in the \vote. 
    Then, the \strongvote is defined as $\langle \vote, B, r, \marker \rangle_i$.
    Each block now contains a \strongqc consisting of $2f+1$ distinct signed \strongvotes of its parent block.
    
    \item {\bf $k$-Endorsements.}
    A \strongvote $\langle \vote, B', r', \marker \rangle_i$ $k$-endorses a block $B$, if and only if $B=B'$, or $B'$ extends $B$ and $\marker<k$.
    A replica is a $k$-endorser of block $B$, if its \strongvote $k$-endorses $B$.

    \item {\bf Strong commit rule.}
    The replica $x$-strong commits a height-$k$ block $B_k$ and all its ancestors, if and only if there exists three adjacent blocks $B_{k-1},B_{k},B_{k+1}$ in the chain with consecutive round numbers, 
    and $B_{k-1},B_{k},B_{k+1}$ all have \underline{at least $x+f+1$ $k$-endorsers}.
\end{itemize}

    \end{mybox}
    \caption{Strengthened Fault Tolerance for Streamlet.}
    \label{fig:streamlet:oft}
\end{figure}

\subsection{Proof of Correctness}
\label{sec:streamlet:proof}
Let $\fc(B)$ denote the set of replicas that certifies $B$ and hence $|\fc(B)|=2f+1$.
\begin{lemma}\label{lem:streamlet:1}
    If a block $B$ with round number $r$ has $E$ $k$-endorsers, then no other conflicting height-$k$ block with round number $r$ can be certified under $t\leq E-f-1$ Byzantine replicas.
\end{lemma}
\begin{proof}
    Suppose for the sake of contradiction, let block $B_k'$ be a conflicting certified height-$k$ block with round number $r$.

    The honest replicas in $\fc(B_k')$ will not vote for $B$, since both blocks are of round number $r$ and honest replicas only vote for one block at each round.
    The replicas in $\fc(B_k')$ may vote for blocks extending $B$, as these blocks may extend the longest certified chain. 
    However, by the definition of $k$-endorse, the honest replicas in $\fc(B_k')$ do not $k$-endorse block $B$, since they voted for $B_k'$ and have $\marker\geq k$. 
    Therefore, the intersection of $\fc(B_k')$ and the set of $k$-endorsers of $B$ contains no honest replicas,  only Byzantine replicas.
    Since the size of intersection is at least $E+|\fc(B_k')|-n=E-f$, there should be at least $E-f$ Byzantine replicas, which contradicts the assumptions that there are at most $t\leq E-f-1$ Byzantine replicas.
    Therefore, any conflicting height-$k$ block $B_k'$ cannot be certified at round $r$.
\end{proof}

\begin{lemma}\label{lem:streamlet:2}
    If a height-$k$ block $B_k$ is $x$-strong committed at some honest replica by the strong commit rule, then no other conflicting height-$k$ block can be certified under $t\leq x$ where $t$ is the number of Byzantine faults.
\end{lemma}

\begin{proof}
    The proof is analogous to the one for Lemma \ref{lem:2}.
    Since block $B_k$ is $x$-strong committed, there exists three blocks $B_{k-1},B_{k},B_{k+1}$ with consecutive round numbers $r-1,r,r+1$, and each of the block $B_{k-1},B_{k},B_{k+1}$ has $\geq x+f+1$ $k$-endorsers.
    Suppose that the lemma is not true for the sake of contradiction and let block $B_k'$ be a conflicting certified height-$k$ block with round number $r'$.
    Consider the following cases.
    \begin{itemize}
        \item When $r'< r-1$.  
        The honest replicas in $\fc(B_k')$ will not vote for $B_{k-1}$, since the certified chain that $B_{k-1}$ extends is shorter than the certified chain that $B_k'$ extends.
        The replicas in $\fc(B_k')$ may vote for blocks extending $B_{k-1}$, as these blocks may extend the longest certified chain. 
        However, by the definition of $k$-endorse, the honest replicas in $\fc(B_k')$ do not $k$-endorse block $B_{k-1}$, since they voted for $B_k'$ and have $\marker\geq k$. 
        Therefore, the intersection of $\fc(B_k')$ and the set of $k$-endorsers of $B_{k-1}$ contains no honest replicas,  only Byzantine replicas.
        Since the size of intersection is at least $E+|\fc(B_k')|-n=E-f$, there should be at least $E-f$ Byzantine replicas, which contradicts the assumptions that there are at most $t\leq E-f-1$ Byzantine replicas.
        
        \item When $r-1\leq r'\leq r+1$.
        By Lemma \ref{lem:streamlet:1}, for $i=-1,0,1$, since $B_{k+i}$ has $x+f+1$ $k$-endorsers, no conflicting height-$k$ block with round number $r+i$ can be certified under $t\leq x$ Byzantine replicas. 
        
        \item When $r'> r+1$. 
        The honest replicas in the set $\fc(B_{k+1})$ have seen a longest certified chain of length $k$ after voting for $B_{k+1}$.
        Thus, they will not vote for block $B_k'$ since it is extending a certified chain of length $k-1$.
        However, replicas in $\fc(B_k')$ may vote for blocks that extend $B_{k+1}$.
        By the definition of $k$-endorse, the honest replicas in $\fc(B_k')$ do not endorse block $B_{k+1}$, since they voted for $B_k'$ and have $\marker\geq k$.
        By the same argument of the $r'<r-1$ case, there should be at least $x+1$ Byzantine replicas, which contradicts the assumptions that there are at most $t\leq x$ Byzantine replicas.
        Therefore, any conflicting height-$k$ block $B_k'$ cannot be certified at round $r' > r+1$.
        
    \end{itemize}
\end{proof}

\begin{theorem}\label{thm:streamlet:safety}
    SFT-Streamlet with the set of $x$-strong commit rules defined in Figure~\ref{fig:streamlet:oft} is safe.
\end{theorem}

\begin{proof}
    Suppose $B$ is $x$-strong committed due to some block $B_k$ being $x$-strong committed directly.
    Suppose for the sake of contradiction, under $t\leq x$ Byzantine replicas, a conflicting block $B'$ is also $x'$-strong committed with $x'\geq x$, due to block $B_{k'}'$ being $x'$-strong committed directly.
    Suppose $B_k$  has height $k$, and $B_{k'}'$ has height $k'$.
    If $k'\geq k$,
    by Lemma \ref{lem:streamlet:2}, no other height-$k$ block other than $B_k$ can be certified, which implies that the conflicting $B_{k'}'$ also cannot be certified. 
    This contradicts the assumption that $B_{k'}'$ is $x'$-strong committed, proving the theorem.
    If $k'<k$, again by Lemma \ref{lem:streamlet:2}, under $t\leq x\leq x'$ faults no other height-$k'$ block other than $B_{k'}'$ can be certified, which implies that the conflicting $B_{k}$ also cannot be certified. 
    This contradicts the assumption that $B_{k}$ is $x$-strong committed, proving the theorem.

\end{proof}

\begin{theorem}\label{thm:streamlet:liveness}
    After GST, if all $c\leq f$ faults are benign (crash) for a period of $n+2$ rounds, then any committed block can be $(2f-c)$-strong committed within at most $n+2$ rounds.
\end{theorem}

\begin{proof}
    The proof is analogous to that of Theorem~\ref{thm:liveness}.
\end{proof}

\subsection{Comparison of SFT-DiemBFT and SFT-Streamlet}
\label{sec:streamlet:comparison}
The solutions for implementing strengthened fault tolerance in DiemBFT and Streamlet share similar intuition and techniques: indirect \strongvotes for blocks in the chain extension can enhance the resilience of the earlier blocks.
As mentioned, to achieve simplicity, Streamlet sacrifices message complexity and latency compared to DiemBFT.
One main difference of Streamlet is the {\em height-based} protocol rules, which consists of proposing/voting/locking/commit rules that depends on the height of the blocks in the blockchain.
One observation is that the height-based protocol rules in Streamlet or SFT-Streamlet add {\em more difficulty for the adversary to launch a long-range attack}, i.e., to strong commit a block on a fork that conflicts with an earlier strong committed block that is buried deep in the blockchain.
Intuitively, since in SFT-Streamlet the honest replicas always vote for the longest certified chain, the adversary must create a fork with a similar length to make honest replicas switch and vote for the fork.
Therefore, if the adversary attempts to create an $x$-strong committed block that conflicts an existing $x$-strong committed block of $h$ heights earlier, it has to corrupt more than $x$ replicas for about $h$ rounds, to create a fork of length $h$ with certified blocks.
In contrast, to conflict an $x$-strong committed block in SFT-DiemBFT, the adversary only needs to corrupt more than $x$ replicas for one round to create a certified block $B'$ in the fork with a higher round number. Then every honest replicas in the original branch will be able to vote for a new block extending $B'$, since $B'$ has a round number larger than their locked value $r_{lock}$. Since the honest replicas can continue to vote for blocks extending $B'$, a conflicting $x$-strong committed block can be created in the fork.

\end{document}